\documentclass[english, 11pt]{article} 
\usepackage{times}  
\usepackage{helvet}  
\usepackage{courier}  
\usepackage{url}  
\usepackage{graphicx}  
\usepackage{algpseudocode}
\usepackage{algorithmicx}

\usepackage{algorithm}

\usepackage[round]{natbib}
\renewcommand{\cite}[1]{\citep{#1}}
\newcommand{\shortcite}[1]{\citeyearpar{#1}}
\newcommand{\citeassubject}[1]{\citet{#1}}

\usepackage[margin = 1in]{geometry}
\usepackage{latexsym}

\usepackage{booktabs}

\usepackage{amsmath, amsthm, amssymb}
\usepackage{accents}
\usepackage{color}
\usepackage{multirow}
\usepackage{algorithm}
\usepackage{algpseudocode}
\usepackage{caption,subcaption}
\usepackage{float}

\newtheorem{definition}{Definition}
\newtheorem{theorem}{Theorem}
\newtheorem{lemma}{Lemma}

\newtheorem{example}{Example}

\newcommand{\report}[1]{{\mathbf{\hat{p}}_{#1}}}
\newcommand{\outcome}[1]{{\mathbf{\tilde{x}}_{#1}}}
\newcommand{\bp}{\mathbf{p}}
\newcommand{\bw}{\mathbf{w}}
\newcommand{\hp}{\mathbf{\hat{p}}}
\newcommand{\tbX}{\tilde{\mathbf{X}}}
\newcommand{\tX}{\tilde{X}}
\newcommand{\tbx}{\tilde{\mathbf{x}}}
\newcommand{\tx}{\tilde{x}}

\ifodd 0
\newcommand{\com}[1]{\textbf{\color{red}(COMMENT: #1)}} 
\newcommand{\clar}[1]{\textbf{\color{green}(NEED CLARIFICATION: #1)}}
\newcommand{\response}[1]{\textbf{\color{magenta}(RESPONSE: #1)}} 
\else

\newcommand{\com}[1]{}
\newcommand{\clar}[1]{}
\newcommand{\response}[1]{}
\fi
\newcommand{\RNum}[1]{\uppercase\expandafter{\romannumeral #1\relax}}
\newcommand{\red}[1]{{\color{black}#1}}
\newcommand{\blue}[1]{{\color{black}#1}}

\title{Randomized Wagering Mechanisms\footnote{This research is based upon work supported in part by the Office of the Director of National Intelligence (ODNI), Intelligence Advanced Research Projects Activity (IARPA), via 2017-17061500006. The views and
conclusions contained herein are those of the authors and should not be interpreted as necessarily representing the official policies, either expressed or implied, of ODNI, IARPA, or the U.S. Government. The U.S. Government is authorized to reproduce and distribute reprints for governmental purposes notwithstanding any copyright annotation therein.}}
\usepackage[T1]{fontenc}
\usepackage{babel}

\author{
  Yiling Chen\\
  Harvard University\\
  {yiling@seas.harvard.edu}
  \and
  Yang Liu\\
  UC Santa Cruz\\
  {yangliu@ucsc.edu}
  \and
  Juntao Wang\\
  Harvard University\\
  {juntaowang@g.harvard.edu}
}

\begin{document}

\maketitle
\begin{abstract}
Wagering mechanisms are one-shot betting mechanisms that elicit agents' predictions of an event. For deterministic wagering mechanisms, an existing impossibility result has shown incompatibility of some desirable theoretical properties. In particular, Pareto optimality (no profitable side bet before allocation) can not be achieved together with weak incentive compatibility, weak budget balance and individual rationality. In this paper, we expand the design space of wagering mechanisms to allow randomization and ask whether there are randomized wagering mechanisms that can achieve all previously considered desirable properties, including Pareto optimality. We answer this question positively with two classes of randomized wagering mechanisms: i) one simple randomized lottery-type implementation of existing deterministic wagering mechanisms, and ii) another family of randomized wagering mechanisms, named \emph{surrogate wagering} mechanisms, which are robust to noisy ground truth. Surrogate wagering mechanisms are inspired by an idea of learning with noisy labels \citep{natarajan2013learning} as well as a recent extension of this idea to the information elicitation without verification setting \citep{DBLP:journals/corr/abs-1802-09158}.  We show that a broad set of randomized wagering mechanisms satisfy all desirable theoretical properties. 

\end{abstract}

\section{Introduction}

Wagering mechanisms \cite{lambert2008self,lambert2015axiomatic,chen2014removing,freeman2017double,freeman2018axiomatic} are one-shot betting mechanisms that allow a principal to elicit participating agents' beliefs about an event of interest without paying out of pocket or incurring a risk. Compared with prediction-market type of dynamic elicitation mechanisms, one-shot wagering \red{is possibly} preferred due to its simplicity. It is particularly designed for agents with immutable beliefs who ``agree to disagree'' and who do not update their beliefs. In a wagering mechanism, each agent submits a prediction for the event and specifies a wager, which is the maximum amount of money that the agent is willing to lose. Then after the event outcome is revealed, the total wagered money will be redistributed among the participants. Researchers have developed wagering mechanisms with various theoretical properties. In particular, \red{\citeauthor{lambert2008self}~\shortcite{lambert2008self,lambert2015axiomatic}} proposed a class of weighted score wagering mechanisms (\textsf{WSWM}) that satisfy a set of desirable properties, including budget balance, individual rationality, incentive compatibility, sybilproofness, among others\red{.\footnote{Definitions of some properties can be found in Section \ref{sec:rwm}.}} \red{\citeassubject{chen2014removing}} later proposed a no-arbitrage wagering mechanism (\textsf{NAWM}) that removes opportunities for participating agents to risklessly profit. 

However, in both \textsf{WSWM} and \textsf{NAWM}, it has been observed that a participant only loses a very small fraction of his total wager even in the worst case. This seems to be undesirable in practice as it is against the \red{spirit} of betting and a wager effectively loses its meaning as a budget. \red{\citeassubject{freeman2017double}} first formalized this observation by indicating that these mechanisms are not Pareto optimal, where Pareto optimality requires that there is no profitable side bet among participants before the allocation of a wagering mechanism is realized. They also proved an impossibility result: Pareto optimality cannot be satisfied together with individual rationality, weak budget balance and weak incentive compatibility. A double clinching auction (\textsf{DCA}) wagering mechanism \citep{freeman2017double} was hence proposed to improve Pareto efficiency. The parimutuel consensus mechanism (\textsf{PCM}) has been shown to satisfy Pareto \red{optimality} \citep{freeman2018axiomatic}, but \red{violates} incentive compatibility.

This paper is another quest of wagering mechanisms with better theoretical properties. We expand the design space of wagering mechanisms to allow randomization on agent payoffs and ask whether we can achieve all aforementioned desirable properties, including Pareto optimality. We give a positive answer to this question: Our randomized wagering mechanisms are the first ones to achieve Pareto optimality along with other properties.

We first show that a simple randomized lottery-type implementation of existing wagering mechanisms (Lottery Wagering Mechanisms (\textsf{LWM})) satisfy all desirable properties. In \textsf{LWM}, instead of receiving re-allocated money from a deterministic wagering mechanism, each agent receives a number of lottery tickets proportional to his payoff in the deterministic wagering mechanism. Then, the agent with the winning lottery wins the total wager (collected from all participants).

We then design another family of randomized wagering mechanisms, the Surrogate Wagering Mechanisms (\textsf{SWM}), by bringing insights from learning with noisy data \citep{natarajan2013learning,scott2015rate} to wagering mechanism design. A \textsf{SWM} first generates a ``surrogate outcome" for each agent according to the true event outcome. An agent's reported prediction is then evaluated using his surrogate but biased outcome together with a bias removal procedure such that in expectation the agent receives a score as if his prediction is evaluated against the true event outcome. Despite being randomized, \textsf{SWM} preserve the incentive properties of a deterministic wagering mechanism. 
We show that \red{certain \textsf{SWM}} satisfy all desirable properties of a wagering mechanism. Notably, \textsf{SWM} are robust to situations where only a noisy copy of the event outcome is available - this property is due to the fact that we borrow the machinery from the literature of learning with noisy data in designing \textsf{SWM}. We believe that this is another unique contribution to the literature of wagering mechanism design.

The rest of this paper is organized as follows. We discuss relevant literature in the Section~\ref{sec:related}. Section \ref{sec:prelim} introduces some preliminaries. We define randomized wagering mechanisms as well as desirable theoretical properties for them in Section \ref{sec:rwm}. Section \ref{sec:lottery} presents a family of lottery-based wagering mechanisms. A family of surrogate wagering mechanisms are introduced in Section \ref{sec:swm}. We extend surrogate wagering mechanisms on \textsf{NAWM} and on the multi-outcome event settings in Section~\ref{sec:general}. Extensive simulations are presented in Section \ref{sec:simulation} to demonstrate the advantages of randomized wagering mechanisms. Section \ref{sec:conclude} concludes this paper. 

\section{Related works}
\label{sec:related}



The ability to elicit \emph{information}, in particular predictions and forecasts about future events, is crucial for many application settings and has been studied extensively in the literature. Proper scoring rules have been designed \cite{Brier:50,Jose:06,Matheson:76,Win:69,gneiting2007strictly} for this purpose, where each agent is rewarded by how well their reported forecasts predicted the true realized outcome (after the outcome is resolved). Later, competitive scoring rule \citep{Kilgour:04} and a parimutuel Kelly probability scoring rule \citep{Johnstone:2007} adapt proper scoring rules to group competitive betting. Both mechanisms are budget balanced so that the principal doesn't need to pay any participant. These spur the further development of the previously discussed wagering mechanisms \citep{lambert2008self,lambert2015axiomatic,chen2014removing,freeman2017double,freeman2018axiomatic} and the examination of their theoretical properties.

Our method used in lottery wagering mechanisms to transfer an arbitrary deterministic wagering mechanism into a randomized one, while maintaining the properties, is inspired by the method proposed in \citep{witkowski2018incentive}. The paper studies the incentive compatible forecasting competition and it transfers scores of multiple predictions into the odds of winning to maintain properties of the scoring rules. \citep{lambert2008self}  proposed a randomization method based on \textsf{WSWM} via randomly selecting strictly proper scoring rules and proper scoring rules with extreme values to increase the stake.  However, this method does not generalize to other deterministic wagering mechanisms.  \citep{cummings2016possibilities} proposed to apply differential privacy technology to randomize the payoff of wagering mechanisms in order to preserve the privacy of each agent's belief. However, their method does not maintain budget balance (in ex-post).

The idea of using randomization in wagering mechanism design is not entirely new, but not thoroughly studied. Both  \citep{lambert2008self, cummings2016possibilities} proposed certain types of randomized wagering mechanisms, but neither of the mechanisms satisfies Pareto optimality. The randomized wagering mechanisms first appeared in~\cite{lambert2008self}. There,  the randomization is restricted to randomly selecting different scoring rules used in \textsf{WSWM}. It introduced this randomization in order to alleviate the the problem that in \textsf{WSWM}, agents only lose a small fraction of their wagers regardless of the event outcome. However, even with this randomization,  an agent won't lose all his wager in the worst when the number of agents is finite. \citep{cummings2016possibilities} applied differential privacy technology to randomize the payoff of wagering mechanisms. Its goal is to preserve the privacy of agents' beliefs. 

Our specific ideas of adding randomness as in the lottery-like wagering mechanisms are inspired by recent works on forecasting competition \citep{witkowski2018incentive}. Our ideas of surrogate wagering mechanisms are inspired by surrogate scoring rules \citep{DBLP:journals/corr/abs-1802-09158},  and the literature on learning with noisy labels \citep{bylander1994learning,natarajan2013learning,scott2015rate}.

\section{Preliminaries}\label{sec:prelim}
In this section, we explain the  scenario where a wagering mechanism applies and formally introduce the deterministic wagering mechanisms.
Consider a scenario where a principal is interested in eliciting subjective beliefs from a set of agents $\mathcal{N}=\{1,2,...,N\}$ about a random variable (event) $X$, which takes a value (outcome) in set $\mathcal{X}=\{0, 1, ..., M-1\}, M\ge 2$. The belief of each agent $i$ is private, denoted as a vector of occurrence probabilities of each outcome $\bp_i=(p_{i}^j)_{j\in\mathcal{X}}\in\Delta^{M-1}$. Following the previous work on wagering mechanism, this paper continues to adopt an immutable belief model for agents. Unlike in a Bayesian model, agents with immutable beliefs do not update their beliefs. The immutable belief model and the Bayesian model are two extremes of agent modeling for information elicitation, with the reality lies in between and arguably closer to the immutable belief side as people do ``agree to disagree.'' Moreover, \citeassubject{lambert2015axiomatic} showed that while \textsf{WSWM} was designed for agents with immutable beliefs, it continued to perform well for Bayesian agents who have some innate utility for trading.

The principal uses a \emph{wagering mechanism} to elicit private beliefs of agents. In a wagering mechanism, each agent reports a probability vector $\hp_i\in\Delta^{M-1}$, capturing his belief, and wagers an amount of money $w_i\in\mathbb{R}_+$. Similar to~\citeassubject{lambert2008self}, we assume that wagers are exogenously determined for each agent and are not a strategic consideration. We use $\hp$ and $\bw$ to denote the reports and the wagers of all agents respectively, and use $\hp_{-i}$ and $\bw_{-i}$ to denote the reports and wagers of all agents other than agent $i$.
\red{In addition, }we use $W_{\mathcal{S}}$ to denote $\sum_{i\in\mathcal{S}}w_i$ for any set of agents $\mathcal{S}\subseteq{\mathcal{N}}$.
After an event outcome $x\in\mathcal{X}$ is realized, the wagering mechanism redistributes all the wagers collected from agents according to $\hp,\bw, x$. The net-payoff of agent $i$ is defined as the payoff or the money that agent $i$ receives from the redistribution minus his wager. A wagering mechanism defines a net-payoff function $\Pi_i(\hp; \bw; x)$ for each agent $i$ with wager constraint $\Pi_i(\hp; \bw; x)\ge -w_i$ and constraint $\Pi_i(\hp; \bw; x)=0$ whenever $w_i=0$. 
The two constraints ensure that no agent can lose more than his wager and no agent with zero wager can gain.

\subsection{Strictly proper scoring rules and weighted score wagering mechanisms}
\emph{Strictly proper scoring rules}~\citep{gneiting2007strictly} are scoring functions proposed and developed to truthfully elicit beliefs from risk-neutral agents. They are building blocks of many incentive compatible wagering mechanisms, such as \textsf{WSWM} and \textsf{NAWM}. 
\red{A strictly proper scoring rule rewards} a prediction $\hp_i$ by a score $s_x(\hp_i)$,  according to the realization $x$ of the random variable $X$. The scoring function $s_x(\cdot)$ is designed such that the expected payoff of truthful reporting is strictly larger than that of any other report, i.e, 
$\mathbb E_{X \sim \mathbf{p}_i}\bigl [s_{X}(\mathbf{p}_i)\bigr] > \mathbb E_{X \sim \mathbf{p}_i}\bigl [s_{X}(\report{i})\bigr],~\forall \report{i} \neq \mathbf{p}_i.$

There is a rich family of strictly proper scoring functions, including Brier scores (for binary outcome event, $s_x(\hat{p}_i)=1-(\hat{p}_i-x)^2$, where $\hat{p}_i$ is agent $i$'s report of $\mathbb P(X=1)$), logarithmic and spherical scoring functions. Strictly proper scoring rules are closed under positive affine transformations.

\emph{\textsf{WSWM}}~\citep{lambert2008self} rewards an agent according to his wager and the accuracy of his prediction relative to that of other agents' predictions. The net-payoff  of agent $i$ in \textsf{WSWM}, is formally defined as 
\begin{equation}
\label{eqn:ws}
\Pi^{\textsf{WS}}_i(\hp;\bw;x) = \frac{w_i W_{\mathcal{N}\backslash \{i\}}}{W_{\mathcal{N}}}\biggl( s_x(\report{i})  
- \sum_{j \in \mathcal{N}\backslash \{i\}} \frac{w_j}{W_{\mathcal{N}\backslash \{i\}}} s_x(\report{j})\biggr),
\end{equation}
where $s_x(\cdot)$ is any strictly proper scoring rule bounded within $[0,1]$. 
\textsf{WSWM} strictly encourages truthful reporting of predictions, because the net-payoff of agent $i$ is a strictly proper scoring rule of his prediction. \red{Meanwhile},   $\sum_{i\in\mathcal{N}}\Pi^{\textsf{WS}}_i$ is always zero by the form of the net-payoff formula, no matter what $s_x(\cdot)$ is. \red{This means that the} budget balance property of Eqn. (\ref{eqn:ws}) doesn't depend on the scoring rules. \red{Our proposed surrogate wagering mechanisms use the same general form of the net-payoff function (but a different scoring rule) to guarantee the ex-post budget balance.}

\section{Randomized wagering mechanisms} 
\label{sec:rwm}
We introduce randomized wagering mechanisms as extensions of deterministic wagering mechanisms. Similar to deterministic wagering mechanisms, the net-payoff of an agent in randomized wagering mechanisms depends on all agents' predictions $\hp$ and wagers $\bw$, as well as the realized outcome $x$. But different from deterministic wagering mechanisms, the net-payoffs are now random variables.
For notational simplicity, we now use $\Pi_i(\hp;\bw;x)$ to represent the random variable of agent $i$'s net-payoff in a randomized wagering mechanism. 
We use $\pi_i(\hp;\bw;x)$ to represent the realization of $\Pi_i(\hp;\bw;x)$. We use $\Pi_i$ and $\pi_i$ as abbreviations when $\hp;\bw;x$ are clear in the context.  We denote the maximum/minimum possible value of a random variable $X$ by $\overline{X}$/$\underline{X}$.
We denote the joint distribution of $\Pi_i(\hp;\bw;x),i\in\mathcal{N}$ by $\mathcal{D}(\hp;\bw;x)$ and the marginal distribution of $\Pi_i(\hp;\bw;x)$ by $\mathcal{D}_i(\hp;\bw;x)$. 

\begin{definition}
Given a set $\mathcal{N}$ of agents, reports $\hp$ and wagers $\bw$ of agents and the event outcome $x$, a \emph{randomized wagering mechanism} 
defines a joint distribution $\mathcal{D}(\hp;\bw;x)$, and pays  agent $i$ by a net-payoff $\Pi_i(\hp;\bw;x)$, where $\Pi_i(\hp;\bw;x), i\in\mathcal{N}$ are jointly drawn from $\mathcal{D}(\hp;\bw;x)$. Moreover, $\underline{\Pi}_i(\hp;\bw;x)\ge -w_i$ and $\Pi_i(\hp;\bw;x)=0$ whenever $w_i=0$.
\end{definition}
A deterministic wagering mechanism is a special case of randomized wagering mechanisms when $\mathcal{D}_i(\hp;\bw;x)$ is a point distribution for all agent $i\in\mathcal{N}$.

\subsection{Desirable properties}\label{sec:property}

In the literature, several desirable properties of wagering mechanisms have been proposed in the deterministic context. 
\red{\citeassubject{lambert2008self}} introduced (\textbf{a}) individual rationality, 
(\textbf{b}) incentive compatibility, 
(\textbf{c}) budget balance, 
(\textbf{d}) sybilproofness, 
(\textbf{e}) anonymity,and
(\textbf{f}) neutrality. \red{\citeassubject{chen2014removing}} introduced (\textbf{g}) no arbitrage. \red{\citeassubject{freeman2017double}}
introduced (\textbf{h}) Pareto optimality. 
We extend these properties to the randomized context. These new properties reduce to the properties defined in the literature for the special case of deterministic wagering mechanisms.

\textbf{(a) Individual rationality} requires that each agent has nothing to lose in expectation by participating.
\begin{definition}
A randomized wagering mechanism is \textbf{individually rational (IR)} if $\forall i,\bp_i,\bw$, and $\hp_{-i}$, there exists $\hp_i$ such that
$$\mathbb E_{X \sim \bp_i, \Pi_i\sim\mathcal{D}_i
(\hp_i,\hp_{-i};\mathbf{w};X)
}\left[ \Pi_i(\hp_i,\hp_{-i};\mathbf{w};X) \right] \ge 0.$$
\end{definition}
\com{ [Fixed:jtw] I didn't change them, but I thought the first input to $\Pi_i$ and $\mathcal{D}_i$ should be $\hp_i$ and not $\bp_i$.}

\textbf{(b) Incentive compatibility} requires that an agent's expected net-payoff is maximized when he reports honestly, regardless of other agents' reports and wagers. 
\begin{definition}
A randomized wagering mechanism is \textbf{weakly incentive compatible (WIC)} if $\forall i, \bp_i, \hp_i\ne \bp_i,\hp_{-i},\mathbf{w}:$
\begin{align*}
\mathbb E_{X \sim \bp_i, \Pi_i\sim\mathcal{D}_i
(\bp_i,\hp_{-i};\bw; X)
}\left[\Pi_i(\bp_i,\hp_{-i};\bw; X)\right] \\
\ge \mathbb
E_{X \sim \bp_i, \Pi_i\sim\mathcal{D}_i
(\hp;\bw; X)
}\left[\Pi_i(\hp;\bw; X)\right].
\end{align*}
A randomized wagering mechanism is \textbf{strictly incentive compatible (SIC)} if the inequality is strict.
\end{definition}

\textbf{(c) Ex-post budget balance} ensures that the principal does not need to subsidize the betting. 
\begin{definition}
A randomized wagering mechanism is \textbf{weakly ex-post budget-balanced (WEBB)} if $\forall \hp,\bw, x:
\sum_{i\in\mathcal{N}}\pi_i(\hp,\bw, x)\le 0$ for any realization  $(\pi_i)_{i\in\mathcal{N}}$ drawn from the joint distribution $\mathcal{D}(\hp,\bw, x)$.
A randomized wagering mechanism is \textbf{ex-post budget-balanced (EBB)} if the equality always holds.
\end{definition}
\textbf{(d) Sybilproofness} requires that no agent can increase its expected net-payoff by creating fake identities and splitting his wager, regardless of other agents' reports and wagers.

\blue{
\begin{definition}
Suppose agent $i$, instead of participating under one account with reported prediction $\hp_{i}$ and wager $w_i$, participates under $k>1$ sybil accounts, with predictions and wagers $\{\hp_{i_l},w_{i_l}\}_{l=1,\dots,k}$ such that $\hp_{i_l}=\hp_{i}, w_{i_l} \geq 0,  \forall l=1,\dots,k$ and $\sum_{l=1}^k w_{i_l}=w_i$. A randomized wagering mechanism is \textbf{sybilproof} if $\forall i, \hp, \bw$,and $x$, and for all sybil reports $\hp_{i_1}, ..., \hp_{i_k}$ and wagers $w_{i_1},...,w_{i_k}$, we have 
\begin{align*}
&\mathbb{E}_ {\Pi \sim\mathcal{D} (\hp;\bw; x)}[\Pi_i(\hp;\bw; x)] 
\\
&\ge\mathbb{E}_{\Pi'\sim\mathcal{D}(\hp';\bw';x)} \bigl[\sum_{l=1}^k\Pi_{i_l}(\hp';\bw';x)\bigr].\end{align*}
where $\hp$, $\bw$ and $\Pi$ are the reports, wagers and net-payoffs when agent $i$ participates under one account and $\hp'$, $\bw'$ and $\Pi'$ are the reports, wagers and net-payoffs when agent $i$ participates using $k$ sybils.
\end{definition}
}

\textbf{(e) Anonymity} requires that agents' identities do not affect their net-payoffs. Let $\sigma_\mathcal{N}$ be a permutation of the set of agents $\mathcal{N}$, and denote $\hp_{\sigma_\mathcal{N}}, \bw_{\sigma_\mathcal{N}}$ the reports and wagers of agents after applying the permutation respectively. Denote $\mathcal{D}_{\sigma_\mathcal{N}}$ the joint distribution of net-payoffs of agents in $\mathcal{N}$ after applying the permutation on agents.
\begin{definition}
A randomized wagering mechanism is \textbf{anonymous} if 
$\forall \sigma_\mathcal{N}, \hp,\bw,x: \mathcal{D}(\hp;\bw;x)=\mathcal{D}_{\sigma_\mathcal{N}}(\hp_{\sigma_\mathcal{N}};\bw_{\sigma_\mathcal{N}};x)$
\end{definition}

\textbf{(f) Neutrality}  requires that the net-payoffs do not depend on the labeling of the event outcomes.
Let $\sigma_\mathcal{M}$ be a permutation of the set of outcomes $\mathcal{M}$. Denote by $\hp_i^{\sigma_\mathcal{M}}$ the reported prediction of agent $i$ after we relabel the outcomes according to permutation $\sigma_\mathcal{M}$, and denote by $\sigma_\mathcal{M}(x)$ the new label of an outcome $x\in\mathcal{M}$.

\begin{definition}
A randomized wagering mechanism is \textbf{neutral} if $\forall \sigma_\mathcal{M}, \hp,\bw,x:$
$$
\mathcal{D}(\hp;\bw;x)=\mathcal{D}(\hp_1^{\sigma_\mathcal{M}},...,\hp_N^{\sigma_\mathcal{M}};\bw;\sigma_\mathcal{M}(x)).
$$
\end{definition}

\textbf{(g) No arbitrage} requires that no agent can risklessly make a profit.
\begin{definition}
A randomized wagering mechanism has \textbf{no arbitrage} if $\forall i, \hp,\bw(\bw>\mathbf{0}),\exists x$ such that $\underline{\Pi}_i(\hp,\bw, x)< 0.$
\end{definition}

\textbf{(h) Pareto optimality} in economics refers to an efficient situation where no trade can be made to improve an agent's payoff without harming any other agent's payoff. In an IR wagering mechanism, agents with different beliefs can always form a profitable (in expectation) wagering game if they all have a positive budget. \red{\citeassubject{freeman2017double}} defined Pareto optimality of a wagering mechanism as a property that agents with different beliefs will each lose all of his wager under at least one of the event outcomes. This ``\red{worst-case}" outcome might be different for different agents. Thus, before  the  event  outcome is  realized, no agent can commit to secure part of his wager from the mechanism and no additional profitable wagering game can be made. We define Pareto optimality for randomized wagering mechanisms in a similar spirit: no agents with different beliefs can commit to secure part of their wagers before the event outcome is realized.
\begin{definition}\label{def:po}
A randomized wagering mechanism is \textbf{Pareto optimal (PO)} if
$\forall \hp,\bw, \forall i, j\in\mathcal{N}$ with $\hp_{i}\ne\hp_{j}, \exists l\in\{i, j\}$ and $x$, such that $\underline{\Pi}_{l}(\hp,\bw,x)=-w_l.$
\com{[Fixed-jtw] Should it be $l\in\{1, 2\}$? Otherwise, I don't understand the index.}
\end{definition}

\begin{table}[t]
  \begin{center}
    \caption{A summary of properties of wagering mechanisms\label{property_table}}
    \label{tab:table1}
    \begin{tabular}
{|c|c|c|c|c|} 
    \hline
      & \textbf{Budget} & \textbf{Incentive} & \textbf{Pareto} &  \textbf{No}  \\
      Mechanism &  \textbf{Balance} & \textbf{Compatibility} & \textbf{Optimality}  & \textbf{Arbitrage}
      \\
      \hline
      \textsf{WSWM}~\cite{lambert2008self} & Strictly&  Strictly & \textbf{False}
     & \textbf{False} \\\hline
      \textsf{NAWM}~\cite{chen2014removing} & \textbf{Weakly} &  Strictly & \textbf{False}  & True \\\hline
      \textsf{DCA}~\cite{freeman2017double} & Strictly &  \textbf{Weakly} & \textbf{False}  & True \\\hline
      \textsf{PCM}~\cite{freeman2018axiomatic} & Strictly &  \textbf{False} & True & True \\\hline
      Randomized \textsf{WSWM}~\cite{lambert2008self} & Strictly &  True & \textbf{False} &   True\\\hline
      Private \textsf{WSWM}~\cite{cummings2016possibilities}& \textbf{False} &  True &  \textbf{False} &  True\\\hline
      \textsf{LWS} (this paper) & Strictly &  True & True & True\\\hline
      \textsf{RP-SWME} (this paper) & Strictly &  True & True & True
      \\\hline
    \end{tabular}
  \begin{tabular}{l}
  (All of the mechanisms in this table satisfy individual rationality, anonymity, neutrality and sybilproofness.)
  \end{tabular}
  \end{center}

\end{table}

\paragraph{Properties of existing wagering mechanisms} We summarize the properties of existing  wagering mechanisms\footnote{\textsf{WSWM}, \textsf{NAWM}, \textsf{DCA}, \textsf{PCM}, randomized \textsf{WSWM} \citep{lambert2008self}, private \textsf{WSWM} \cite{cummings2016possibilities}} and ours in Table~1. No existing mechanism satisfies all properties (a)-(h). Moreover, \red{\citeassubject{freeman2017double}} showed an \textbf{impossibility result} that for deterministic wagering mechanisms, it is impossible to achieve properties IR, WIC, WEBB, and PO simultaneously. For existing randomized wagering mechanisms, the randomized \textsf{WSWM} in \citep{lambert2008self} only satisfies PO in the limit of large population of participants, and the private \textsf{WSWM}~\citep{cummings2016possibilities} does not satisfy WEBB and PO.

\section{Lottery wagering mechanisms}\label{sec:lottery}

In this section we introduce a family of randomized wagering mechanisms, the \emph{lottery wagering mechanisms} (LWM), which extends arbitrary deterministic wagering mechanisms into randomized wagering mechanisms. We will show that LWM easily preserve (the randomized version of) the  properties of the underlying deterministic wagering mechanisms, while achieving Pareto optimality, overcoming the impossibility result. 

In lottery wagering mechanisms, each agent receives a number of lottery tickets in proportion to the \emph{payoff} he gets under a deterministic wagering mechanism, and a winner is drawn from all the lottery tickets to win the entire pool of wagers. The mechanisms are designed in a way such that the expected payoff of each agent is equal to his payoff in the underlying deterministic wagering mechanisms and each agent has a positive probability to lose all his wager. Hence, no profitable side bet exists and the mechanisms are Pareto optimal. 
We formally present the lottery wagering mechanism that extends an arbitrary deterministic wagering mechanism \textsf{DET}  in Mechanism~\ref{ag_lws}. To distinguish the payoff from the net-payoff, we denote the payoff of agent $i$ by $\pi'_i$ .
\begin{small}
\begin{algorithm}[h]
 \floatname{algorithm}{Mechanism}
\caption{Lottery Wagering Mechanisms}\label{ag_lws}
\begin{algorithmic}[1]
\State Compute the payoff of each agent $i$ under a \textsf{DET}: $\pi_i' \leftarrow w_i+\Pi_i(\hp;\bw;x)$.
\State 
Each agent has winning probability $\frac{\pi'_j}{\sum_{i\in\mathcal{N}}\pi'_j}$. Draw a lottery winner $i^*\in\mathcal{N}$.
\State Winner $i^*$ is assigned a net-payoff $\sum_{i\in\mathcal{N}\backslash\{i^*\}} w_i$ and any agent $j\ne i^*$ has a net-payoff $-w_j$.
\end{algorithmic}
\end{algorithm}
\end{small}

Lottery wagering mechanisms are powerful in obtaining desirable theoretical properties. 
We show in Theorem~\ref{thm_lwm} that the lottery wagering mechanism that \red{extends \textsf{WSWM}}, namely Lottery Weighted Score wagering mechanism (\textsf{LWS}), satisfies all properties (a)-(h).
\begin{theorem}
\label{thm_lwm}
\textsf{LWS} satisfies all properties (a) - (h).
\end{theorem}

\red{We notice that although \textsf{LWS} satisfies all desirable properties, it can be unsatisfying because (1) agents have high variance in payoff and (2) except the winning agent, all other agents lose money. }
\red{To alleviate these issues, we can  mix \textsf{LWS} with \textsf{WSWM}} by assigning each of them a probability to be executed. \red{The resulting mechanism still satisfies all the properties (a)-(h). The probabilistic mixture allows us to adjust the variance of the payoffs as well as agents' winning probabilities in the resulting mechanism.}

\section{Surrogate wagering mechanisms}\label{sec:swm}

In this section, we propose the \emph{surrogate wagering mechanisms (\textsf{SWM})}. We first introduce the generic \textsf{SWM}, then a  variant of \textsf{SWM} that achieves the desirable theoretical properties and at the same time have moderate variance in payoffs and higher winning probabilities for accurate predictions. We then notice that randomization opens up the possibility of dealing with situations where only noisy ground truth is available. We discuss how to extend our results to this noisy setting.

\subsection{Generic surrogate wagering mechanisms}
A surrogate wagering mechanism consists of three main steps: (1) \textsf{SWM} generates a surrogate event outcome for each agent based on the true event outcome and a randomization device; (2) \textsf{SWM} evaluates each agent's prediction according to the surrogate event outcome using a designed scoring function such that the score is an unbiased estimate of the score derived by applying a strictly proper scoring rule to the ground truth outcome; (3) \textsf{SWM} applies \textsf{WSWM} to the scores based on the surrogate event outcome to determine the final net-payoff of each agent. Next, we explain these three steps in details. For clarity and simplicity of exposition, we consider only binary events, i.e., $\mathcal{X}=\{0,1\}$, in this section. Extension to multi-outcome events will be introduced later.
\vspace{-0.5em}
\paragraph{Step 1. Surrogate event outcomes} A \textsf{SWM} generates a surrogate event outcome $\tilde{X_i}$ for each agent $i\in\mathcal{N}$. Denote $\tbX=(\tX_1, \tX_2, ..., \tX_N)$. $\tX_i$'s are drawn independently conditional on $X$, and are specified by  \textsf{SWM}. The conditionally marginal distribution $\mathbb P(\tX_i|X), i\in\mathcal{N}$ can be expressed by two parameters, the error rates of the surrogate outcome: $e_1^i=\mathbb P(\tX_i=0|X=1)$ and $e_0^i=\mathbb P(\tX_i=1|X=0)$. The conditionally marginal distribution $\mathbb P(\tX_i|X)$ can be any distribution satisfying $\forall i\in\mathcal{N}:e_1^i+e_0^i\ne 1$.\footnote{\red{When $e_0+e_1=1$, $\tX_i$ turns out to be independent with $X$, and thus provides no information about $X$. We thus exclude $e_1^i+e_0^i=1$.} }
We use $\tbx$ and $\tx_i$ to denote the realization of $\tbX$ and $\tX_i$ respectively.
\vspace{-0.5em}
\paragraph{Step 2. Computing unbiased scores} Given a strictly proper scoring rule $s_x(\cdot)$ within [0,1], \textsf{SWM} computes the score of agent $i$ as \red{$\varphi\circ s_{\tx_i}(\hat{p}_i)$}, where 
\small{
\begin{align}
\varphi \circ s_{\tilde{x}_i}(\hat{p}_i) = \frac{(1-e^i_{1-\tilde{x}_i})s_{\tilde{x}_i}(\hat{p}_i) - e^i_{\tilde{x}_i}s_{1-\tilde{x}_i}(\hat{p}_i) }{1-e^i_0-e^i_1}.\label{def:phi}
\end{align}
}
$\tx_i$ is the realized surrogate event outcome for agent~$i$.
Lemma~\ref{lemma:unbias} shows that \red{$\varphi$} is an unbiased operator on the score $s_{\tx_i}(p_i)$ in the sense that $\mathbb E_{\tX_i|x}[\varphi \circ s_{\tX_i}(\hat{p}_i)] = s_{x}(\hat{p}_i).$

\begin{lemma}[Lemma 3.4 of \citep{DBLP:journals/corr/abs-1802-09158}]
$\forall x\in\{0,1\},\forall \hat{p}_i,e^i_0, e^i_1\in[0,1]$ and $ e^i_0+e^i_1\ne 1$, we have
$\mathbb E_{\tX_i|x}[\varphi \circ s_{\tX_i}(\hat{p}_i)] = s_{x}(\hat{p}_i).
$
\label{lemma:unbias}
\end{lemma}\com{[fixed-jtw] In this lemma, should it be that the sum of the error rates not equal to 1? Same for the corollary below.}

\blue{Lemma~\ref{lemma:unbias} implies that if $s_{x}(\hat{p}_i)$ is a strictly proper scoring rule, then $\varphi\circ s_{\tx_i}(\hat{p}_i)$ is also a strictly proper scoring rule.}

\paragraph{Step 3. Computing net-payoffs} In the final step,  \textsf{SWM} computes the net-payoff of agent $i$ using \textsf{WSWM} and the unbiased score of agent $i$, i.e., replacing score $s_x(\hat{p}_i)$ in Eqn.~(\ref{eqn:ws}) by score $\varphi\circ s_{\tilde{x}_i}(\hat{p}_i)$. Formally, we have
\small{
\begin{align}\label{def:sam}
\Pi^{\textsf{SWM}}_i(\hp,\mathbf{w},x) 
= \frac{w_i W_{\mathcal{N}\backslash \{i\}}}{W_{\mathcal{N}}}\biggl( \varphi \circ s_{\tilde{x}_i}(\hat{p}_i) 
- \sum_{j \in \mathcal{N}\backslash \{i\}} \frac{w_j}{W_{\mathcal{N}\backslash \{i\}}} \varphi \circ s_{\tilde{x}_j}(\hat{p}_j)\biggr),
\end{align}
}
where $x$ and $\tx_i,i\in\mathcal{N}$ are the event outcome and the surrogate event outcome for each agent $i$ respectively. 

\begin{algorithm}[t]
 \floatname{algorithm}{Mechanism}
\caption{ Surrogate Wagering Mechanisms}\label{swm}
\begin{algorithmic}[1]
\State Collect the predictions $\report{}$ and wagers $\mathbf{w}$.
\State Select error rate $e_0^i, e_1^i\in[0,1]$ and $e_0^i+e_1^i\ne 1,  \forall i$.
\State Generate surrogate outcome $\tilde{X}_i,\forall i$ such~that
$\mathbb P(\tX_i=1|X=0)=e_0^i,\, \mathbb P(\tX_i=0|X=1)=e_1^i$.
\State Score each agent $i\in\mathcal{N}$ according to Eqn.~(\ref{def:phi}).
\State Pay each agent $i\in\mathcal{N}$ a net-payoff using Eqn.~(\ref{def:sam}).
\end{algorithmic}
\end{algorithm}

We formally present \textsf{SWM} in Mechanism~\ref{swm}. According to Lemma~\ref{lemma:unbias} (applying to each score terms), we have $\forall i, x,\hp,\bw:$
$\mathbb{E}_{\Pi_i^{\textsf{SWM}} \sim\mathcal{D}(\hp;\bw;x)}[\Pi_i^{\textsf{SWM}}(\hp;\bw;x)]=\Pi^{\textsf{WS}}_i(\hp;\bw;x).$
Because the deterministic \textsf{WSWM} satisfies properties ((a)-(f)) \citep{lambert2008self},\com{ [fixed--jtw]cite Lambert et al. here}   
\textsf{SWM} also satisfies these properties. A realization of the score \red{$\varphi\circ s_{\tX_i}(p_i)$} can be larger than 1, implying that agent $i$ can lose (or win) more than what he can lose (or win) in the deterministic \textsf{WSWM}. However, we also notice that for some extreme values of error rates, the constraint $\underline{\Pi}_i(\hp;\bw;x)\ge -w_i$ can be violated\footnote{For example, in a wagering game, two agents both wager 1 and report $1$ and $0$, respectively. Let $s_x(\hat{p}_i)=1-(x-\hat{p}_i)^2, e^i_j=0.4,i=1,2,j=0,1$. In the worst case of agent 1, the surrogate outcomes are realized as $\tx_1=0, \tx_2=1$. Then, $\pi_1 = -5<-1$.}, i.e., an agent may lose more than \red{his} wager, which makes \textsf{SWM} invalid. In the next section, we show that by selecting error rates in a subtle way, we can obtain all the properties (a)-(h) without violating the wager constraint $\underline{\Pi}_i(\hp;\bw;x)\ge -w_i$.

\subsection{\textsf{SWM} with Error rate selection (\textsf{SWME}) and random partition \textsf{SWME} (\textsf{RP-SWME})}\com{What's the second S in SWS represent? (What's SWS short for?) The acronym SWS is not very informative. Also, I'm confused whether the second mechanism is called random partition SWM or random partition SWS. Should the SWM here be SWS?}
We notice that according to Lemma~\ref{lemma:unbias}, no matter which error rates $e_0$, $e_1$ are chosen, the unbiasedness property of SWM holds, i.e., $\mathbb{E}_{\Pi_i \sim\mathcal{D}(\hp;\bw;x)}[\Pi_i^{\textsf{SWM}}(\hp;\bw;x)]=\Pi^{\textsf{WSWM}}_i(\hp;\bw;x)$. In other words, we can choose the error rates in an arbitrary way (even depending on $\hp, \bw$) without changing the expected net-payoff\footnote{The expectation is taken over the randomness of the mechanism conditioned on the event outcome.} of each agent under any realized event outcome.  
This gives us the flexibility to tune the maximum amount of money each agent can win or lose in the game, while preserving the properties ((a)-(f)) inherited from \textsf{WSWM}.  

Given reports $\hp$ and wagers $\bw$ but not the event outcome $x$, the error rate pair that guarantees no wager violation under any outcome $x\in\mathcal X$ and any realization of the randomness induced by \textsf{SWM} may not be unique. We propose Algorithm~\ref{ers} to select a pair of error rates $e_0, e_1$ after the reports and wagers are collected, \com{ex-post is confusing. It often refers to after the event outcome is realized. What we mean here is after the reports and wagers are collected.} such that at least one agent loses all his wager in the worst case w.r.t. the outcome and the randomness of \textsf{SWM}. We name the mechanism as \textsf{SWME} when we use Algorithm~\ref{ers} to select the error rates for  \textsf{SWM}.

\begin{algorithm}[h]
\caption{ Error Rate Selection Algorithm}\label{ers}
\begin{algorithmic}[1]
\State Collect the predictions $\report{}$ and wagers $\mathbf{w}$.
\State $\forall i$: $s_i^w\leftarrow\min_{x\in\mathcal{X}} s_x(\hat{p}_i), s_i^b\leftarrow\max_{x\in\mathcal{X}} s_x(\hat{p}_i)$. 

\State For each agent $i\in\mathcal{N}$, compute $r_i$: 
$r_i~\leftarrow~\frac{1}{2}~+\frac{(1-\frac{w_i}{W_{\mathcal{N}}})(s_i^w-s_i^b)+\sum_{j\in\mathcal{N}\backslash\{i\}}\frac{w_j}{W_{\mathcal{N}}}(s_j^w-s_j^b)}{2(2+s_i^w+s_i^b-\sum_{j\in\mathcal{N}}\frac{w_j}{w_\mathcal{N}}(s_j^w+s_j^b))}$

\State If $\min_{j\in\mathcal{N}} \{r_j\}=0.5$, set $e^i_1 = e^i_0 = 0, \forall i$, else set $e^i_1 = e^i_0 = \min_{j\in\mathcal{N}} \{r_j\}, \forall i$.
\end{algorithmic}
\end{algorithm}

\begin{lemma}\label{sws:er}
\textsf{SWME} has no wager violation and when there exists at least one report $\hat{p}_i \neq 0.5$, at least one of the agents loses all his wager in the worst case w.r.t. the event outcome and the randomness of \textsf{SWME}.
\end{lemma}
\begin{proof}
In this proof, we use Brier Score as the scoring rule used by the mechanism, i.e., $s_x(\hat{p}_i)=1-(x-\hat{p}_i)^2$, and $\hat{p}_i$ is agent $i$'s report of $\mathbb{P}(X=1)$. The proof can be extended to other strictly proper scoring rule within [0, 1].

We first consider the corner case where all agents reports 0.5. It can be verified that in Algorithm 2, $\min_{i\in\mathcal{N}} r_i=0.5$, and the algorithm sets $e_0^i=e_1^i=0,\forall i$ and \textsf{SWME} is reduced to \textsf{WSWM}. Thus, no wager violation happens.

Next, we consider the scenario that $\exists i\in\mathcal{N}, \hat{p}\ne 0.5$. In this scenario, we first prove that, in Algorithm 2 $\forall i, r_i\in(0,0.5)$.

We have  $\forall i, s_i^w, s_i^b\in[0,1], s_i^w\le s_i^b$ (the equality only holds when $\hat{p}_i=0.5$), $s_i^w+s_i^b\in [0.5, 1]$. 
Let 
$$A=(1-\frac{w_i}{W_\mathcal{N}})s_i^w-\sum_{j\in\mathcal{N}\backslash\{i\}}\frac{w_j}{W_\mathcal{N}}s_j^b$$ 
and 
$$B=(1-\frac{w_i}{W_\mathcal{N}})(s_i^w+s_i^b)-\sum_{j\in{\mathcal{N}\backslash\{i\}}}\frac{w_j}{W_\mathcal{N}}(s_j^w+s_j^b).$$
We have $r_i=\frac{1}{2}+\frac{2A-B}{2(2+B)}$, $A>-1, B\in (-1,1)$ and $2A-B = \sum_{j\in \mathcal{N}\backslash\{i\}}\frac{w_j}{W_\mathcal{N}}(s_j^w-s_j^b)+(1-\frac{w_i}{W_{\mathcal{N}}})(s_i^w-s_j^b)> 0$ (there exists at least one agent $i\in\mathcal{N}$ that $\hat{p}_i\ne 0.5$). Therefore, $\frac{2A-B}{2+B}\in(-1, 0)$. We have $r_i=\frac{1}{2}+\frac{2A-B}{2(2+B)}\in(0, 0.5)$.

Next, we prove that if let $r_i$ be a variable, and let $e_0^i=e_1^i=r_i$, the worst cast net-payoff $\pi^w_i$ (w.r.t. the event outcome and the randomness of the mechanism) of agent $i$ is a decreasing function of $r_i$.

In the worst case of agent $i$, $\varphi \circ s_{\tilde{x}_i}(\hat{p}_i) = \frac{(1-r_i)s_i^w-r_is_i^b}{1-2r_i}, \varphi \circ s_{\tilde{x}_j}(\hat{p}_j) = \frac{(1-r_i)s_i^b-r_is_i^w}{1-2r_i}$ and $\pi^w_i = w_i\frac{(A-Br_i)}{1-2r_i}$. We have $\frac{\partial \pi_i^w}{\partial r_i} = w_i\frac{2A-B}{(1-2r_i)^2}<0$. Therefore, $\pi_i^w$ is decreasing with $r_i$.

Finally, it is easy to verify that when $r_i=\frac{1}{2}+\frac{2A-B}{2(2+B)}$, $\pi_i^w=-w_i$.

Therefore, when we set for each agent $i\in\mathcal{N}$, $e_i^0=e_i^1=\min_{j\in\mathcal{N} }r_j$, no agent can lose more than his wager and agent $i^*=\text{argmin}_{j\in{\mathcal{N}}} r_j$ loses all his wager in the worst case. 
\end{proof}

Note Lemma \ref{sws:er} does not imply PO for \textsf{SWME} - if there exist two agents who have different predictions and have wager left even in their own worst cases, they can form a profitable bet against each other. We propose a variant of \textsf{SWME} to fix this caveat as follows.

\paragraph{Random partition \textsf{SWME} (\textsf{RP-SWME})} Lemma \ref{sws:er} implies that when agents are partitioned into groups of two, there will not exist side bets. Meanwhile, a smaller number of agents imposes less restrictions in selecting the error rates, and thus each agent's wager can be fully leveraged in the randomization step. We would like to note that this is a very unique property of \textsf{SWME}: as both shown in \citeauthor{freeman2017double}~\shortcite{freeman2017double} and our experimental results, when the number of agents is small,  existing wagering mechanisms (including \textsf{DCA}) all have low risk, i.e., have only a small portion of wager to lose in the worst case. This not only implies that \textsf{SWME} is particularly suitable for small group wagering  but also points out a way to further improve the risk property of \textsf{SWME}, i.e. via randomly partitioning agents into smaller groups. We formally present the random partition \textsf{SWME} in Mechanism~\ref{rpsWSWM}. 
We show in next section that \textsf{RP-SWME} achieves all properties (a)-(h).

\begin{algorithm}[h]
 \floatname{algorithm}{Mechanism}

\caption{ Random Partition \textsf{SWME} (\textsf{RP-SWME})}\label{rpsWSWM}
\begin{algorithmic}[1]
\State Partition agents into groups of two. If $N$ is odd, leave one group with three agents. \com{[Fixed-jtw] I don't understand the at most one group with three agents part. Why is 3 so special? Can I have groups with smaller than 3 agents?}
\State Run \textsf{SWME} for each group.
\end{algorithmic}
\end{algorithm}

\subsection{Properties of \textsf{SWME} and \textsf{RP-SWME}}

\begin{theorem}\label{thm:sws}
Both (\textsf{SWME}) and (\textsf{RP-SWME})  satisfy properties (a)-(g). (\textsf{RP-SWME})  satisfies (h).\com{ [fixed-jtw] Didn't we just show that SWME is PO? Why SWME doesn't satisfy h?}
\end{theorem}
\begin{proof}
We prove the properties one by one.

\paragraph{(a) Individual rationality and (b) (strictly) incentive compatibility:}

First consider \textsf{SWME}. For an arbitrary profile of reports $\hp$ and wagers $\bw$, Algorithm~\ref{ers} outputs a profile $\mathcal E$ of error rates of all agents. 
Denote by $\hat{\varphi}^i_{\mathcal E}(\cdot)$ the corresponding surrogate function specified using the error rate profile $\mathcal E$ for agent $i$. For each $i$ and $j \in \mathcal N$:
\begin{align*}
\mathbb E_{X\sim p_i,\tilde{X}_j}\bigl [\hat{\varphi}^j_{\mathcal E} \circ  s_{\tilde{X}_j}(\hat{p}_j)\bigr] = & 
p_i \mathbb E_{\tilde{X}_j|X=1}[\hat{\varphi}^j_{\mathcal E} \circ s_{\tilde{X}_j}
(\hat{p}_j)]
+ (1-p_i) \mathbb E_{\tilde{X}_j|X=0}[\hat{\varphi}^j_{\mathcal E} \circ s_{\tilde{X}_j}(\hat{p}_j)]\\
=&p_i \cdot s_{X = 1}
(\hat{p}_j) + (1-p_i)\cdot s_{X = 0}
(\hat{p}_j) = \mathbb E_{X\sim p_i}[s_{X}(\hat{p}_j)],
\end{align*}
using Lemma~\ref{lemma:unbias}. \red{Then, using the linearity of expectation, we have (here $\tilde{\mathbf{X}}$ encodes the randomness in $\Pi^{\textsf{SWME}}_i$)}
\begin{align*}
\mathbb E_{X\sim p_i,\tilde{\mathbf{X}}}\bigl[\Pi^{\textsf{SWME}}_i(\hp,\mathbf{w},X)\bigr]
=& \frac{w_i W_{\mathcal{N}\backslash \{i\}}}{W_{\mathcal{N}}}\biggl( \mathbb E_{X\sim p_i,\tilde{X}_i}[\hat{\varphi}^i_{\mathcal E} \circ s_{\tilde{X}_i}(\hat{p}_i)] \nonumber 
- \sum_{j \in \mathcal{N}\backslash \{i\}} \frac{w_j}{W_{\mathcal{N}\backslash \{i\}}} \mathbb E_{X\sim p_i,\tilde{X}_j}[\hat{\varphi}^j_{\mathcal E} \circ s_{\tilde{X}_j}(\hat{p}_j)]\biggr)\\
=& \mathbb E_{X\sim p_i}\biggl[\frac{w_i W_{\mathcal{N}\backslash \{i\}}}{W_{\mathcal{N}}}\biggl( s_{X}(\hat{p}_i) - \sum_{j \in \mathcal{N}\backslash \{i\}} \frac{w_j}{W_{\mathcal{N}\backslash \{i\}}}  s_{X}(\hat{p}_j)\biggr)\biggr]\\
=&\mathbb E_{X\sim p_i}\bigl[ \Pi^{\text{WS}}_i(\hp,\mathbf{w},X) \bigr]~.
\end{align*}

Note the above holds for any possible reports ($\forall \mathcal E$). Thus the incentive properties, i.e., individual rationality and strictly incentive compatibility of \textsf{WSWM} will preserve. The proof for  \textsf{RP-SWME} is similar, with the only difference in that each agent's net-payoff is further averaged over the random partitions (but IR and SIC under each possible partition). 

\paragraph{(c)  Ex-post budget balance:} This can be shown via writing down the sum of net-payoffs defined in Eqn. (\ref{def:sam}). \com{[fixed-jtw] Should we refer to equation 3 since we have it?}Our note below Eqn. (\ref{eqn:ws}) also states that the budget balance property doesn't depend on the specific forms of the scoring functions therein.  We formally present the deduction as follows:

\begin{align*}
\sum_i \Pi^{\text{SWME}}_i(\hat{p}_i, w_i, \cdot) =& \sum_i \frac{w_i W_{\mathcal N\backslash \{i\}}}{W_{\mathcal N}}\biggl( \varphi \circ s_{\tilde{x}_i}(\hat{p}_i) - \sum_{j \in \mathcal N\backslash \{i\}} \frac{w_j}{W_{\mathcal N\backslash \{i\}}}{W_{\mathcal N}} \cdot \varphi \circ s_{\tilde{x}_j}(\hat{p}_j)\biggr)\\
=&\sum_i \biggl( \frac{w_i W_{\mathcal N\backslash \{i\}}}{W_{\mathcal N}} \varphi \circ s_{\tilde{x}_i}(\hat{p}_i)  - \sum_{j \neq i} \frac{w_j W_{\mathcal N\backslash \{j\}}}{W_{\mathcal N}} \cdot  \frac{w_i}{W_{\mathcal N \backslash \{j\}}}{W_{\mathcal N}} \cdot \varphi \circ s_{\tilde{x}_i}(\hat{p}_i) \biggr)\\
=&\sum_i \biggl( \frac{w_i W_{\mathcal N\backslash \{i\}}}{W_{\mathcal N}} \varphi \circ s_{\tilde{x}_i}(\hat{p}_i)  - \frac{w_i W_{\mathcal N\backslash \{i\}}}{W_{\mathcal N}} \varphi \circ s_{\tilde{x}_i}(\hat{p}_i) \biggr)\\
= & 0.
\end{align*}

The above also shows that for each group from the random partition of (\textsf{RP-SWME}), ex-post budget balance is satisfied. Thus, we also proved ex-post budget balance for (\textsf{RP-SWME}).

\paragraph{(d) Sybilproofness:}  In \textsf{RP-SWME}, any pair of agents with different beliefs have a positive probability to be partitioned into a sub-group.  Applying Lemma~\ref{sws:er}, at least one of them loses all his wager in the worst case. Thus,  by Definition~\ref{def:po}, \textsf{RP-SWME} is PO.
\begin{lemma}
\label{LM:SYB}
If a (randomized) wagering mechanism $\mathcal W$ is (weakly) budget-balanced, (weakly) incentive compatible, Sybilproof, then the mechanism $\mathcal W*$ that first uniformly randomly pairs agents in groups of two and then runs mechanism $\mathcal W$ for each group is still Sybilproof.
\end{lemma}

\begin{proof}
We prove the claim for the case that an agent is only allowed to create two identities. The claim holds in general, as we can alway merge two identities into one without decreasing the payoff, following the result of the case of two. 

Fixing an arbitrary belief $\mathbf{p}_i$ of agent $i$, we denote the $E_i^{\mathcal W}(\report{},\mathbf{w}):=\mathbb{E}_{X\sim\mathbf{p}_i, \mathcal{D}^{\mathcal W}(\report{},\mathbf{w},X)}[\Pi_i(\report{},\mathbf{w},X=x)] $, where $\mathcal{D}^{\mathcal W}(\cdot)$ is the distribution specified by mechanism $\mathcal W$. 
Suppose an agent $i$ divides its wager $w_i$ into two wagers   $w_{i1}, w_{i2}$, and reports two predictions $\report{i1}, \report{i2}$ correspondingly. 
We have $\forall \report{i1}, \report{i2}, w_{i1}, w_{i2}, \report{-i}, \mathbf{w}_{-i}, x$, 
\begin{small}
\begin{align*}
& E_i^{\mathcal W*}(\report{i1}, \report{i2},  \report{-i}, w_{i1}, w_{i2}, \mathbf{w}_{-i})  \\
=  &  \sum_{j\ne i} \frac{1}{N}E_i^{\mathcal W}(\report{i1}, \report{j}, w_{i1}, w_j) +\sum_{j\ne i} \frac{1}{N}E_i^{\mathcal W}(\report{i2}, \report{j}, w_{i2}, w_j) 
+
 \frac{1}{N}\left(E_{i1}^{\mathcal W}(\report{i1}, \report{i2}, w_{i1}, w_{i2}) + E_{i2}^{\mathcal W}(\report{i1},\report{i2},w_{i1}, w_{i2})\right) \\
\le & \sum_{j\ne i} \frac{1}{N}E_i^{\mathcal W}(\report{i1}, \report{j}, w_{i1}, w_j) +\sum_{j\ne i} \frac{1}{N}E_i^{\mathcal W}(\report{i2}, \report{j}, w_{i2}, w_j) 
\,\,\,\,\,\,\,\,\,\,\,\,\,\,\,\,\,\,\,\,\,\,\,\,\,\,\,\,\,\,\, 
(\mathcal W\text{ is (weakly) budget balance})\\
\le & \sum_{j\ne i} \frac{1}{N}E_i^{\mathcal W}(\mathbf{p}_i, \report{j}, w_{i1}, w_j) +\sum_{j\ne i} \frac{1}{N}E_i^{\mathcal W}(\mathbf{p}_i, \report{j}, w_{i2}, w_j) 
\,\,\,\,\,\,\,\,\,\,\,\,\,\,\,\,\,\,\,\,\,\,\,\,\,\,\,\,\,\,\,\,\,\,\,\,
 (\mathcal W\text{ is (weakly) incentive compatible})\\
\le & \sum_{j\ne i} \frac{1}{N}E_i^{\mathcal W}(\mathbf{p}_i, \report{j}, w_{i}, w_j) 
\,\,\,\,\,\,\,\,\,\,\,\,\,\,\,\,\,\,\,\,\,\,\,\,\,\,\,\,\,\,\,\,\,\,\,\,\,\,\,\, \,\,\,\,\,\,\,\,\,\,\,\,\,\,\,\,\,\,\,\,\,\,\,\,\,\,\,\,\,\,\,\,\,\,\,\,\,\,\,\, 
\,\,\,\,\,\,\,\,\,\,\,\,\,\,\,\,\,\,\,\,\,\,\,\,\,\,\,\,\,
 (\mathcal W\text{ is sybilproof})\\
\le &  \sum_{j\ne i} \frac{1}{N-1}E_i^{\mathcal W}(\mathbf{p}_i, \report{j}, w_{i}, w_j) =  E_i^{\mathcal W*}(\mathbf{p}_i, \report{-i}, w_{i}, \mathbf{w}_{-i})
\end{align*}
\end{small}
Therefore,  $\mathcal W*$ is sybilproof.
\end{proof}

~\\

\paragraph{(e) Anonymity:} For \textsf{SWME}, this proof can follow from the fact that the randomness (error rate selection) in \textsf{SWME} and t\textsf{RP-SWME} depends only on the reports and wagers of agents and do not depend on the identities of agents and the fact that the expected net-payoffs of agents are the same with those of \textsf{WSWM} (Corollary 1), which is anonymous~\cite{lambert2008self}. \textsf{RP-SWME} only adds a random partition of agents in \textsf{SWME} and the partition does not depend on the identities of agents. Thus, \textsf{RP-SWME} is also anonymous.
~\\

\paragraph{(f) Neutrality:} For \textsf{SWME}, this proof can follow from the fact that the randomness (error rate selection) in \textsf{SWME} and t\textsf{RP-SWME} depends only on the reports and wagers of agents and do not depend on the labeling of the outcomes and the fact that the expected net-payoffs of agents are the same with those of \textsf{WSWM} (Corollary 1), which is neutral~\cite{lambert2008self}. \textsf{RP-SWME} only adds a random partition of agents in \textsf{SWME} and the partition does not depend on the labeling of the outcomes. Thus, \textsf{RP-SWME} is also neutral.

~\\

\paragraph{(g) Non-arbitrage opportunity:} Now we prove that \textsf{SWME} does not allow arbitrage opportunity. The idea is simple and straight-forward: fix the set of prediction $\mathbf{p}_{-i}$ and wagers $\mathbf{w}$. First notice the fact that under each possible realization $\tilde{x}_i$, $\tilde{x}_{-i}$ can be any possible realizations. Since $s_{\tilde{X}_i=1}(p_i)$ and $s_{\tilde{X}_i=0}(p_i)$ have opposite monotonicity, we know there does not exist an interval for risklessly predictions. 

The above non-arbitrage opportunity is \emph{ex-post}, but the arbitrage opportunity persists when agents evaluate the conditional expectation of his score with respect to the random flipping step (which is the same as \textsf{WSWM}), which remains a concern when each agent participates in multiple event forecasts. This concern will be resolved when we apply the idea of surrogate wagering to the non-arbitrage wagering mechanism (\textsf{NAWM}). For details please refer to Section~\ref{sec:snawm}~~.

For \textsf{RP-SWME}, it runs \textsf{SWME} on each pair of agents after the random partition. Therefore, agents also have no arbitrage opportunity.

\paragraph{(h) Pareto optimality:} In \textsf{RP-SWME}, any pair of agents with different beliefs have a positive probability to be partitioned into a sub-group.  Applying Lemma~\ref{sws:er}, at least one of them loses all his wager in the worst case. Thus,  by Definition~\ref{def:po}, \textsf{RP-SWME} is PO.
\end{proof}
\subsection{Wager with noisy ground truth} 
The above method also points out a way to implement a wagering mechanism with a noisy ground truth, as \textsf{SWM} is able to remove the noise in outcomes in expectation. The ability to wager with noisy ground truth provides informative information to agents who participated in a wagering mechanism immediately only when a noisy copy of outcome is available. 
We present the key idea below, while not re-defining all properties w.r.t. $\hat{X}$ instead of $X$ - the changes are rather straight-forward.

Suppose we know a noisy estimate $\hat{X}$ on $X$, and denote the error rate of $\hat{X}$ as $\hat{e}_1,\hat{e}_0$ (which we know, and agents trust us in knowing these two numbers), we will be able to reproduce our surrogate wager mechanism by plugging $\hat{X},\hat{e}_1,\hat{e}_0$ into Eqn. (\ref{def:phi}),\com{Not a good reference as the equation doesn't have the error rates explicitly stated.} if we ignore the PO property for now. We similarly will have the wager violation issue pointed out earlier - we however do not have the control of the error rates directly. An easy fix is via the following affine transformation of the wagering scores: suppose under the worst case, the random flipping will incur $-{scale} \cdot w_i$\com{Should the negative sign be removed?} wager score (net-payoff) with ${scale} > 1$. We can then rescale every agent's wager score by $1/\textsf{scale}$. Note the above affine transformation does not affect the incentive and other properties of the original surrogate wagering mechanism, as
$
\mathbb E\bigl[\varphi \circ \Pi^{\text{WS}}_i(\cdot))\bigr] =\frac{1}{\textsf{scale}} \cdot \mathbb E\bigl[\Pi^{\text{WS}}_i(\cdot))\bigr] 
$.\footnote{We didn't apply the scaling in \textsf{SWME} when there exists other options, as the scaling will effectively decrease the expected payment of each agent.} 
To achieve PO, we can further random partition agents into groups of two and flip on $\hat{X}$ according to certain error rates $\hat{e}_0^i, \hat{e}_1^i$ for each agent $i$. Let $\tX_i$ be the flipped outcome. We can establish the error rates of $\tX_i$ w.r.t. the ground truth $X$ and $\hat{e}_0^i, \hat{e}_1^i$ by following equations:
\begin{align*}
\mathbb P(\tilde{X}_i = 1| X=0) 
& = \sum_{x\in\{0,1\}}\mathbb P(\tilde{X}_i = 1, \hat{X} = x| X=0)\\
& =\sum_{x\in\{0,1\}} \mathbb P(\tilde{X}_i = 1| \hat{X} = x, X=0) \cdot \mathbb P(\hat{X}=x|X=0)  \\
& = \hat{e}^i_0 \cdot (1-\hat{e}_0)+(1-\hat{e}^i_1) \cdot \hat{e}_0,
\end{align*}
and similarly $\mathbb P(\tilde{X}_i = 0| X=1) =\hat{e}^i_1 \cdot (1-\hat{e}_1)+(1-\hat{e}^i_0) \cdot \hat{e}_1$.
It's easy to see that when $\hat{e}_1+\hat{e}_0 \neq 1$, we can tune the error rates of $\tilde{\mathbf{X}}$ via tuning $\hat{e}^i_1,\hat{e}^i_0$. This step corresponds to the error selection step in \textsf{SWME}, i.e., Algorithm~\ref{ers}.

\section{Extensions of \textsf{SWM}}
\label{sec:general}
We discuss a couple of useful extensions of \textsf{SWM}: i). one is instead of building on $\textsf{WSWM}$, we show the idea of surrogate idea can also build upon another deterministic wagering mechanism $\textsf{NAWM}$. (ii). We extend our results to a multi-outcome setting. 

\subsection{Surrogate \textsf{NAWM}}\label{sec:snawm}

We note that the bias removal procedure adopted in \textsf{SWM} does not rely the specific underlying wagering mechanism heavily. We demonstrate the idea with a non-arbitrage wagering mechanism (\textsf{NAWM}, \citep{chen2014removing})\footnote{Though the randomization device already grants us the non-arbitrage property, we pick this mechanism for i. its simplicity for presentation, as \textsf{NAWM} also extends from \textsf{WSWM}. ii. we will show in experiments later that we empirically observe higher risk when applying this surrogate based randomized \textsf{NAWM}. }. 

%

Notice that since $\Pi^{\text{NA}}_i(\cdot)$ is not linear in the surrogate scores of each agent, the budget balance argument is not as easy as in the \textsf{WSWM} case. Nonetheless we notice the following fact proved in \citep{chen2014removing}:
\[
\Pi^{\text{NA}}_i(\hat{p}_i,\mathbf{\hat{p}}_{-i},\mathbf{w},X=x) = \Pi^{\text{WS}}_i(\hat{p}_i,\mathbf{\hat{p}}_{-i},\mathbf{w},X=x) -  \Pi^{\text{WS}}_i(\bar{\hat{p}}_i,\mathbf{\hat{p}}_{-i},\mathbf{w},X=x)
\]
where $\bar{\hat{p}}_i$ denotes the average prediction from $j \neq i$. Then we can safely apply the surrogate idea to the first \textsf{WSWM} scoring term:
\[
\varphi \circ \Pi^{\text{NA}}_i(\hat{p}_i,\mathbf{\hat{p}}_{-i},\mathbf{w},\tilde{\mathbf{X}}=\outcome{}) = \varphi \circ \Pi^{\text{WS}}_i(\hat{p}_i,\mathbf{\hat{p}}_{-i},\mathbf{w},\tilde{\mathbf{X}}=\outcome{})-  \Pi^{\text{WS}}_i(\bar{\hat{p}}_i,\mathbf{\hat{p}}_{-i},\mathbf{w},X=x)
\]
This mechanism will enjoy the higher risk property introduced by surrogate wagering, as well as the non-arbitrage (in conditional expectation) brought in by \textsf{NAWM}.

\subsection{Multi-outcome events}
\label{non-binary}

For simplicity, our previous discussions focused largely on the binary outcome scenario. As promised, we now show that our results extend to the non-binary events. Recall that there are $M$ outcomes, denoting as $[0,1,2,...,M-1]$. Denote the following confusion matrix
\[ C = 
\begin{bmatrix}
    c_{0,0}       & c_{0,1} &  \dots & c_{0,M-1} \\
    c_{1,0}        & c_{1,1}  & \dots & c_{1, M-1} \\
    \hdotsfor{4} \\
    c_{M-1,0}       & c_{M-1,1} & \dots & c_{M-1,M-1}
\end{bmatrix}
\]
and each entries $c_{j,k}$ indicates the flipping probability for generating a surrogate outcome:
$
c_{j,k} = \Pr[\tilde{X}_i = k|X=j].
$

The core challenge of this extension is to find an unbiased operator $\varphi$. Writing out the conditions for unbiasedness (s.t. $\mathbb E_{\tilde{X}_i|x}[\varphi \circ s_{\tilde{X}_i=\tilde{x}_i}(\mathbf{\hat{p}})] = s_{x}(\mathbf{\hat{p}}).
$), we need to solve the following set of functions to obtain $\varphi(\cdot)$ (short-handing $\varphi\circ s_{x}(\mathbf{\hat{p}})$ as $\varphi_x(\mathbf{\hat{p}})$):
\begin{align*}
s_0(\mathbf{\hat{p}})& = c_{0,0} \cdot \varphi_0(\mathbf{\hat{p}}) + c_{0,1}  \cdot \varphi_1(\mathbf{\hat{p}}) + \dots +c_{0,M-1}  \cdot \varphi_{M-1}(\mathbf{\hat{p}})\\
s_1(\mathbf{\hat{p}}) &= c_{1,0}  \cdot \varphi_0(\mathbf{\hat{p}}) + c_{1,1}  \cdot \varphi_1(\mathbf{\hat{p}}) + \dots +c_{1,M-1}  \cdot \varphi_{M-1}(\mathbf{\hat{p}})\\
&....\\
s_{M-1}(\mathbf{\hat{p}})& = c_{M-1,0}  \cdot \varphi_0(\mathbf{\hat{p}}) + c_{M-1,1}  \cdot \varphi_1(\mathbf{\hat{p}}) + \dots +c_{M-1,M-1}  \cdot \varphi_{M-1}(\mathbf{\hat{p}})
\end{align*}
Denote by $\mathbf{s}(\mathbf{\hat{p}}) = [s_0(\mathbf{\hat{p}});s_1(\mathbf{\hat{p}});...;s_{M-1}(\mathbf{\hat{p}})]$, and $\mathbf{\varphi}(\mathbf{\hat{p}}) = [\varphi_0(\mathbf{\hat{p}});\varphi_1(\mathbf{\hat{p}});...;\varphi_{M-1}(\mathbf{\hat{p}})]$. Then the above equation becomes equivalent with the following system of equation:
$\mathbf{s}(\mathbf{\hat{p}})  = C \cdot \mathbf{\varphi}(\mathbf{\hat{p}}). 
$ Choose a $C$ with full rank. For instance when $M>2$ we can set $\forall j$,
$
c_{j,j} = \frac{1}{2},~ c_{j,k} = \frac{1}{2(M-1)}, ~k \neq j
$ - not hard to verify that such a $C$ is indeed full rank. Then we are ready to solve for $\mathbf{\varphi}(\mathbf{p})$ as follows:
\begin{align}
\mathbf{\varphi}(\mathbf{\hat{p}}) = C^{-1}\cdot \mathbf{s}(\mathbf{\hat{p}}). \label{unbias:multi}
\end{align}
With defining above unbiased surrogate operator, all other discussions generalize fairly straight-forwardly - such a $\varphi$ will give us the same equation as established in the lemma below for the non-binary event outcome setting:
\begin{lemma}
Define $\varphi(\cdot)$ as in Eqn. (\ref{unbias:multi}), and flip $\tilde{X}_i$ using $C,x$. Then $\mathbb E_{\tilde{X}_i|x}[\varphi \circ s_{\tilde{X}_i=\tilde{x}_i}(\mathbf{\hat{p}})] = s_{x}(\mathbf{\hat{p}}).
$
\end{lemma}
We include a detailed example of $\varphi$ for three-outcome events below.
%

\paragraph{Example of $\varphi$ for three-outcome events}

\begin{example}
An example with $M=3$. Suppose we flip the outcome using the uniform-error confusion matrix:
\[ C = 
\begin{bmatrix}
    0.5      & 0.25 &  0.25 \\
    0.25       & 0.5  & 0.25 \\
    0.25       & 0.25  & 0.5 
\end{bmatrix} \Rightarrow C^{-1} = 
\begin{bmatrix}
   3     & -1 &  -1 \\
    -1       & 3  & -1 \\
    -1       & -1  & 3 
\end{bmatrix}
\]
Therefore we obtain a closed-form of $\varphi$:
\begin{align*}
\mathbf{\varphi}_0 (\mathbf{\hat{p}})&= 3  \mathbf{s}_0(\mathbf{\hat{p}}) - \mathbf{s}_1(\mathbf{\hat{p}}) - \mathbf{s}_2(\mathbf{\hat{p}}) \\ 
\mathbf{\varphi}_1(\mathbf{\hat{p}}) &= - \mathbf{s}_0(\mathbf{\hat{p}}) +3 \mathbf{s}_1(\mathbf{\hat{p}}) - \mathbf{s}_2(\mathbf{\hat{p}}) \\ 
\mathbf{\varphi}_2(\mathbf{\hat{p}}) &= -  \mathbf{s}_0(\mathbf{\hat{p}}) - \mathbf{s}_1(\mathbf{\hat{p}}) + 3\mathbf{s}_2(\mathbf{\hat{p}}) 
\end{align*}
\end{example}
\section{Evaluation}\label{sec:simulation}
In this section, we evaluate \textsf{LWS} and \textsf{RP-SWME} with extensive simulations. We first compare the efficiency of \textsf{LWS} and \textsf{RP-SWME} with that of other existing deterministic (weakly) incentive compatible mechanisms \textsf{WSWM}, \textsf{NAWM} and \textsf{DCA}. The results show that the two randomized wagering mechanisms outperform the three deterministic wagering mechanism. Then, we compare the variance of payoff and the probability of winning money within the two randomized wagering mechanisms. The results show that \textsf{RP-SWME} is better than \textsf{LWS} in these two matrices. 

\subsection{Simulation Setup}
We simulate both the binary events and the multi-outcome events.
For binary events, we generated six sets of agents' predictions and wagers according to the combinations of three different prediction models and two different wager models. With a little abuse of notation, we denote that an event happens with probability $q$ and that agent $i$ believes that the event to predict will happen with  probability $p_i$ and will not happen with probability $1-p_i$. We use three models to generate predictions $p_i,i\in\mathcal{N}$:
\begin{enumerate}
\item Uniform model: For each event, $p_i$ is independently drawn from a uniform distribution over $[0,1]$.
\item Logit-Normal model:  This model assumes that $p_i$, when being mapped to the real line by a logit function as $\log\left(\frac{p_i}{1-p_i}\right)$, is independently drawn from a Normal distribution $\mathcal{N}(\log(\frac{q}{1-q})^{1/\alpha}, \sigma^2)$, i.e., $p_i\sim \text{Logit-Normal}\left(\log(\frac{q}{1-q})^{1/\alpha}, \sigma^2\right)$.  $q, \alpha, \sigma^2$ are model parameters. This model is proposed and used to estimate the  happening probability of the event in~\cite{satopaa2014combining}, where $q$ is regarded as an estimator of the happening probability and $\alpha$ models the under-confident effect on human forecasters. Based on a real prediction dataset over 1300 forecasters and 69 geopolitical events collected in~\cite{satopaa2014combining}, this model outperforms most existing models to estimate the happening probability of events, which leads us to believe this model a good alternate to generate prediction data. In our simulations, we adopted $\alpha = 2$, which best fits the aforementioned real prediction dataset, $\sigma^2=1$, and $q$ is drawn uniformly from $[0,1]$ for each event. 

\item Synthetic model: this synthetic model is introduced from a set of simulation studies in \cite{ranjan2010combining,allard2012probability,satopaa2014combining}. 
 The model assumes that the happening probability of an event to be predicted by  $N$  is given by $q = \Phi(\sum_{i=1}^N u_i )$, where $\Phi$ is the cumulative distribution function of a standard normal distribution and $u_i$ is independently drawn from $\mathcal{N}(0,1)$. Each agent knows the true probability generating model and $u_i$ but not $u_j, \forall j\ne i$. Accordingly, each agent's calibrated belief of the happening probability of the event is given by $p_i = \Phi(\frac{u_i}{\sqrt{2N-1}})$.
\end{enumerate}

We use two models to generate the wagers of agents: 
 \begin{enumerate}
\item Uniform model: All agents' wagers are equal to 1.
\item Pareto model:  This model assumes that the wager $w_i$ of agent $i$ follows  the Pareto distribution, which is often adopted to model the distribution of wealth in a population. In the simulations of \cite{freeman2017double}, the authors selected the shape parameter and scale parameter of the Pareto distribution as 1.16 and 1 correspondingly, which is the distribution depicted as  ``20\% of the population has 80\% of the wealth''. We adopted the same parameters for comparison purpose. 
\end{enumerate}

For events with multiple outcomes, we simulated three sets of data with the number of possible outcomes 3, 6, 9 correspondingly. In each set, we drew the predictions from uniform distribution over the whole probability space and drew the wagers according to the Uniform model.

\subsection{Comparison of efficiency of wagering mechanisms}

\begin{figure}[t]
\centering
\begin{subfigure}{0.32\textwidth}
\includegraphics[width=\textwidth]{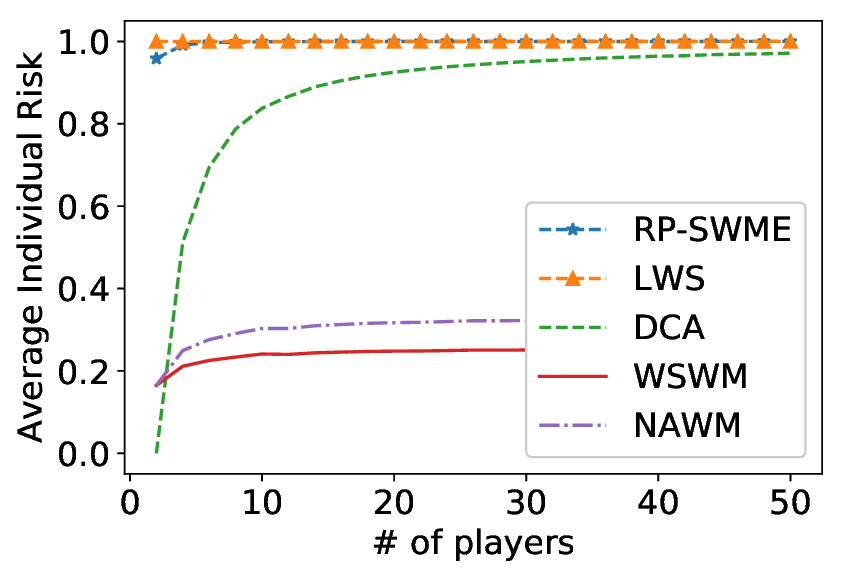}
\caption{\label{risk_u_u}$\mathbf{p}\sim$ Uniform, $\mathbf{w}\sim$ Uniform}
\end{subfigure}
\begin{subfigure}{0.32\textwidth}
\includegraphics[width=\textwidth]{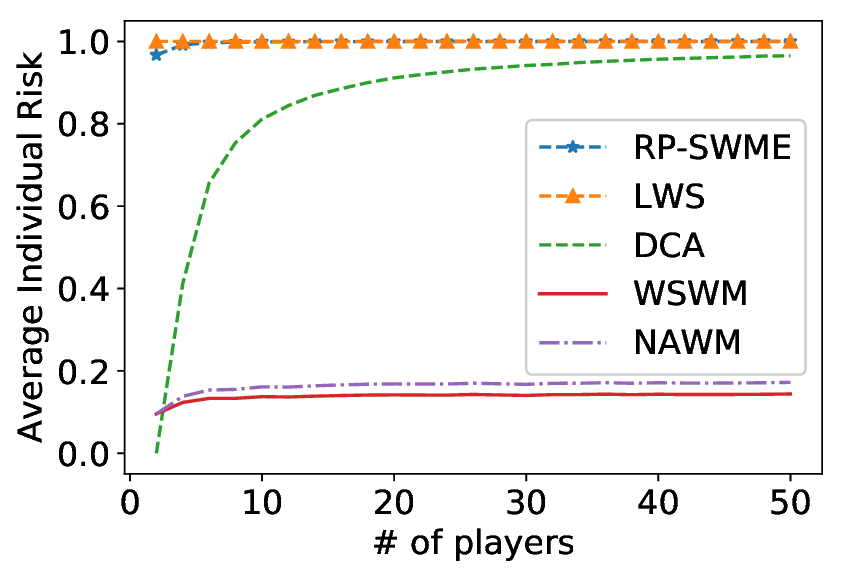}
\caption{\label{risk_u_l}$\mathbf{p}\sim$ Logit, $\mathbf{w}\sim$ Uniform}
\end{subfigure}
\begin{subfigure}{0.32\textwidth}
\includegraphics[width=\textwidth]{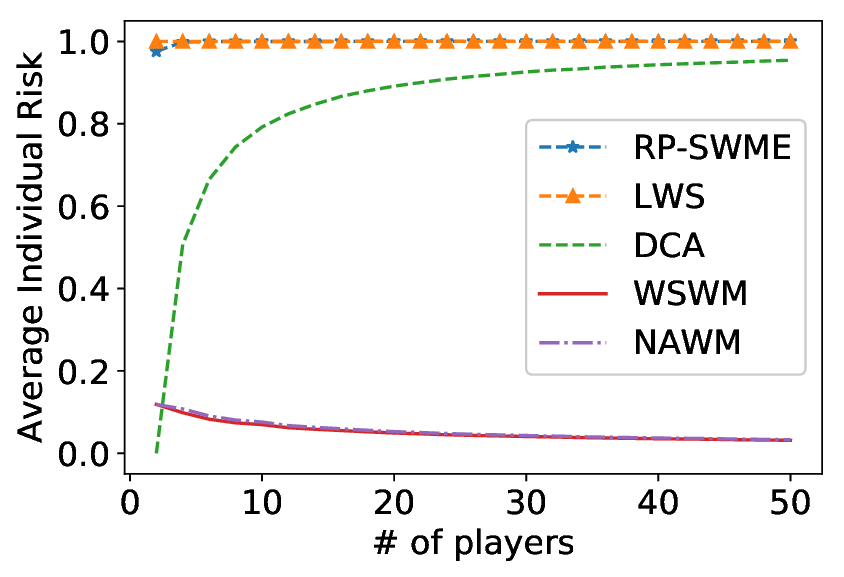}
\caption{\label{risk_u_s}$\mathbf{p}\sim$ Synthetic, $\mathbf{w}\sim$ Uniform}
\end{subfigure}
\begin{subfigure}{0.32\textwidth}
\includegraphics[width=\textwidth]{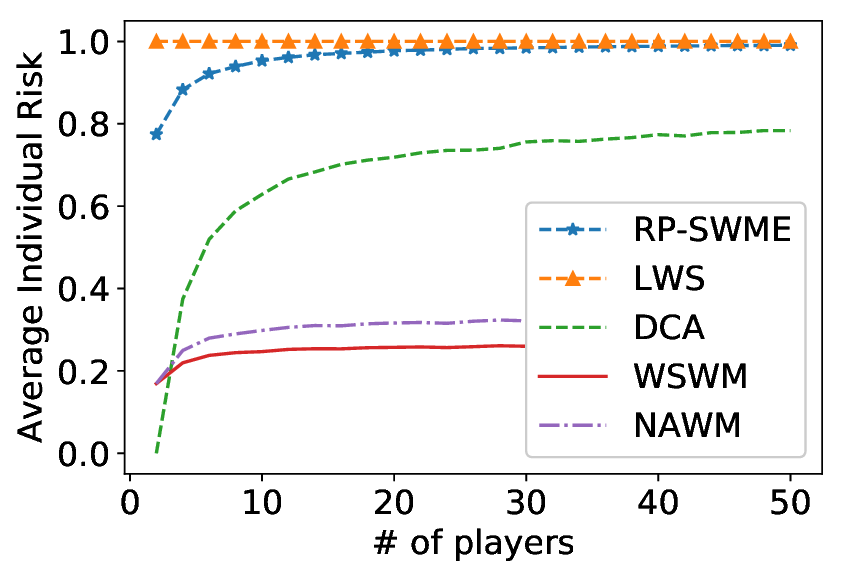}
\caption{$\mathbf{p}\sim$ Uniform, $\mathbf{w}\sim$ Pareto}
\label{risk_p_u}
\end{subfigure}
\begin{subfigure}{0.32\textwidth}
\includegraphics[width=\textwidth]{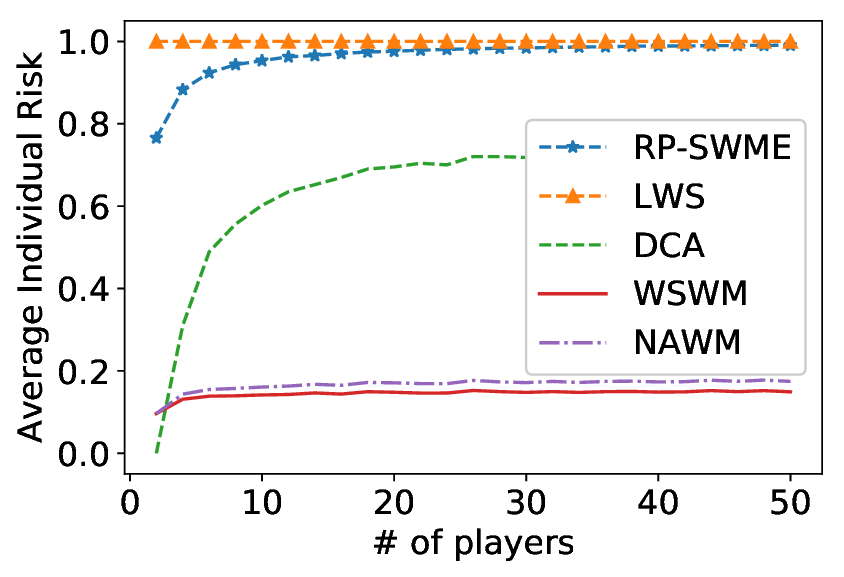}
\caption{$\mathbf{p}\sim$ Logit, $\mathbf{w}\sim$ Pareto}
\label{risk_p_l}
\end{subfigure}
\begin{subfigure}{0.32\textwidth}
\includegraphics[width=\textwidth]{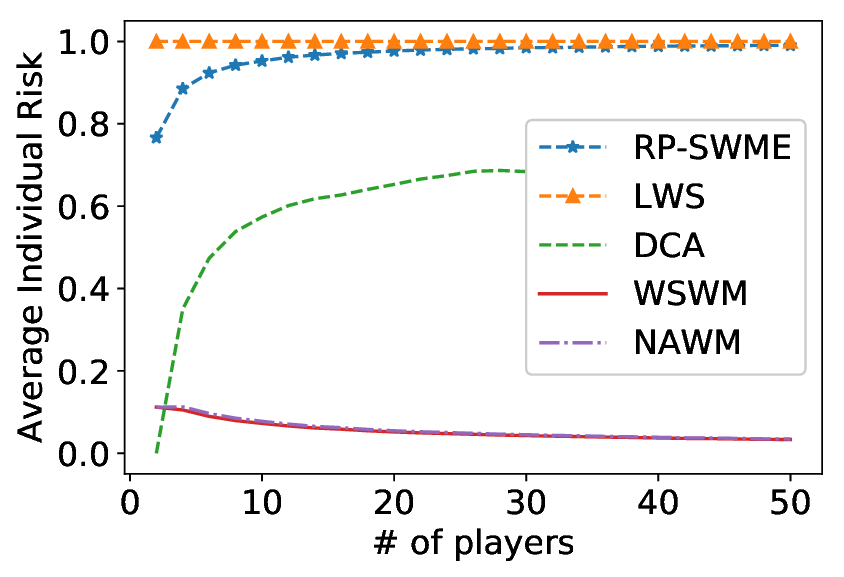}
\caption{$\mathbf{p}\sim$ Synthetic, $\mathbf{w}\sim$ Pareto}
\label{risk_p_s}
\end{subfigure}
\caption{\label{risk_all}Average individual risk of each of five wagering mechanisms as a function of $N$ under different prediction and wager models}
\end{figure}
\begin{figure}[t]
\centering
\begin{subfigure}{0.32\textwidth}
\includegraphics[width=\textwidth]{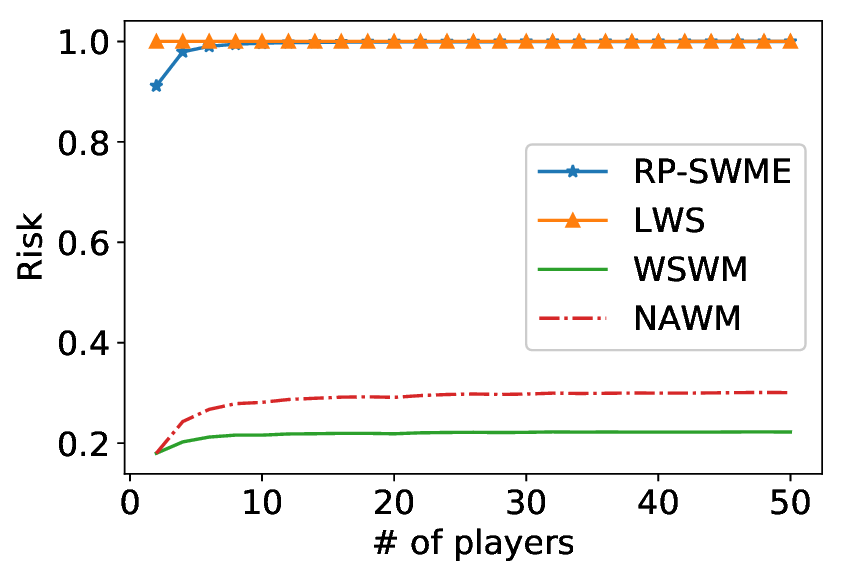}
\caption{\label{risk_multi_3}Events with 3 outcomes}
\end{subfigure}
\begin{subfigure}{0.32\textwidth}
\includegraphics[width=\textwidth]{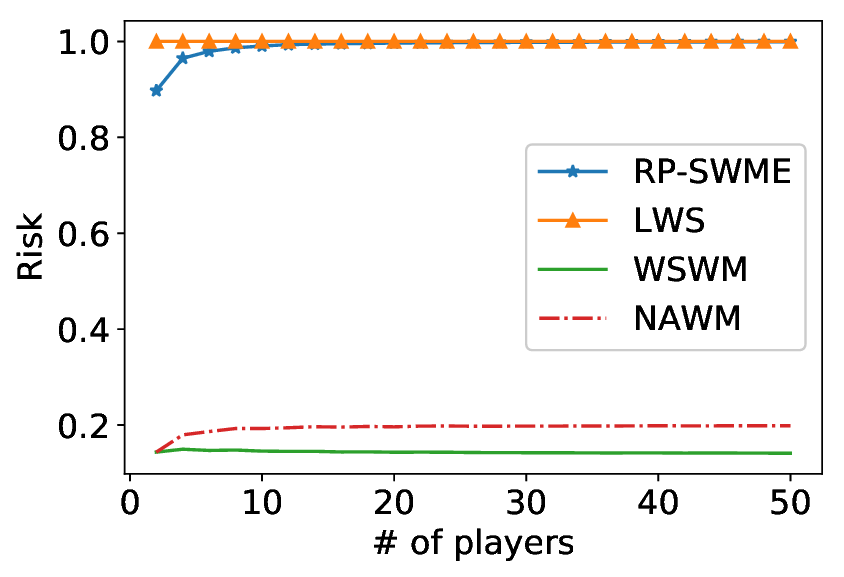}
\caption{\label{risk_multi_6}Events with 6 outcomes}
\end{subfigure}
\begin{subfigure}{0.32\textwidth}
\includegraphics[width=\textwidth]{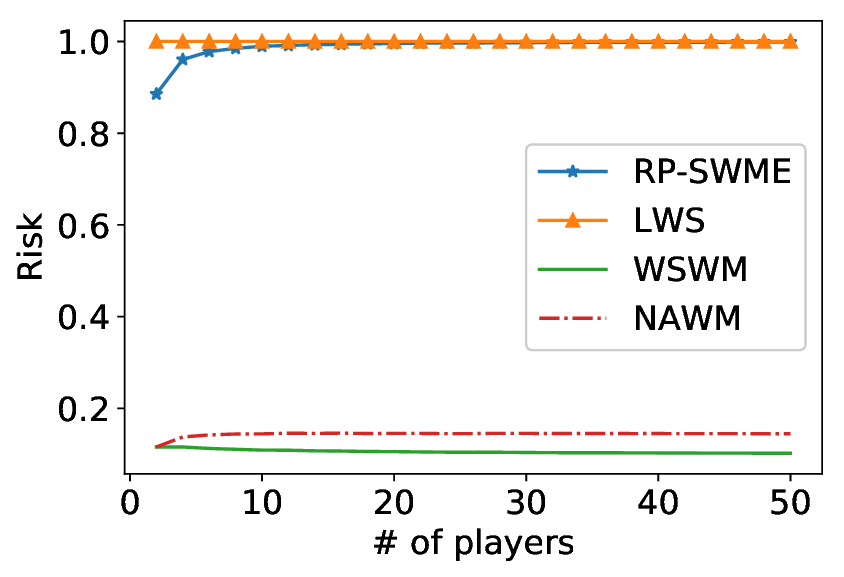}
\caption{\label{risk_multi_9}Events with 9 outcomes}
\end{subfigure}
\caption{\label{risk_multi_all}Average individual risk of each of four mechanisms under events with multiple outcomes}
\end{figure}
We show that \textsf{LWS} and  \textsf{RP-SWME} are more efficient than existing deterministic (weakly) incentive compatible mechanisms \textsf{WSWM}, \textsf{NAWM} and \textsf{DCA}. We evaluate the efficiency by two metrics: \emph{Average individual risk} and \emph{Average money exchange rate}.

\emph{Individual risk} is the percent of wager that an individual agent can lose in the worst case w.r.t. the event outcome and the randomness of the mechanisms. The average individual risk is an indicator of Pareto optimality, because the average individual risk equal to 1 (i.e., no one can commit to secure a positive wager before the wagering game) is a sufficient condition of Pareto optimality.
\emph{Money exchange rate} is the total amount of money exchanged  in the game after the outcome of a wagering mechanism is realized, divided by the total amount of wagers. Average money exchange rate measures the efficiency of an average wagering game.

In our simulations, we vary the number of agents for 2 to 50 with a step of 2. For each number of agents, we randomly generate 1000 events and the agents' predictions and wagers for each of the six combinations of prediction models and wager models, and take the average of individual risk  and money exchange rate over the 1000 events. When calculating the money exchange, we use the expectation of the money exchange over all possible outcomes according to the happening probability of each outcome. This happening probability is either specified in the model generating the predictions, or otherwise, drawn from a uniform distribution over the corresponding probability space. 

In the simulations, both \textsf{RP-SWME} and \textsf{LWS} achieve the highest average individual risk (approximately 1)  under all conditions (\# of outcomes, \# of agents, prediction models, and wager models) we simulated (Figure~\ref{risk_all}, \ref{risk_multi_all}). 
In contrast, the best of the deterministic mechanisms \textsf{DCA}, only achieves an approximate 1 average individual risk  when the wagers of agents are uniform and the number of participants is more than 30 (Figure~\ref{risk_u_u}-\ref{risk_u_s}). Its average individual risk drops to 0.6 when the wagers of agents follows the Pareto distribution (Figure~\ref{risk_p_u}-\ref{risk_p_s}). This result shows that the two randomized mechanisms effectively remove the opportunity for side bet and take use of all the wagers before the outcome is realized.

\begin{figure}[t]
\centering
\begin{subfigure}{0.32\textwidth}
\includegraphics[width=\textwidth]{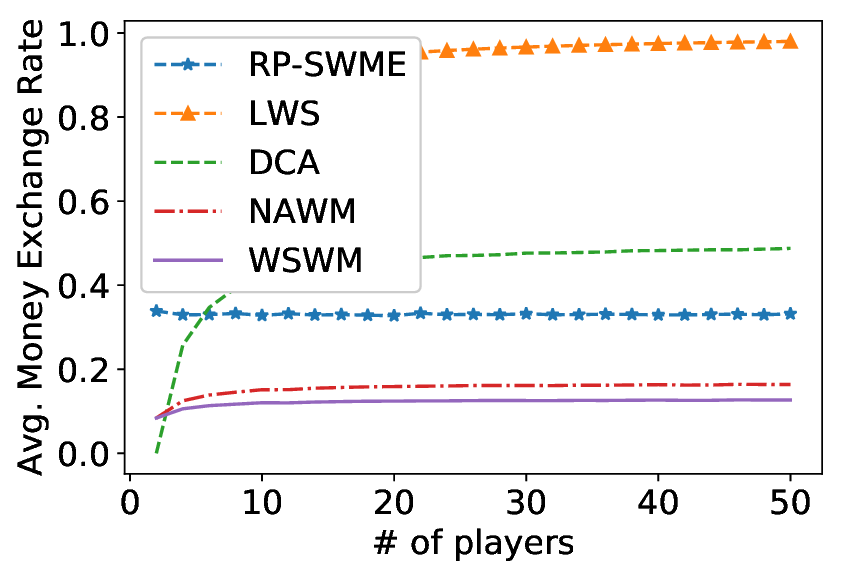}
\caption{$\mathbf{p}\sim$ Uniform, $\mathbf{w}\sim$ Uniform}
\label{me_u_u}
\end{subfigure}
\begin{subfigure}{0.32\textwidth}
\includegraphics[width=\textwidth]{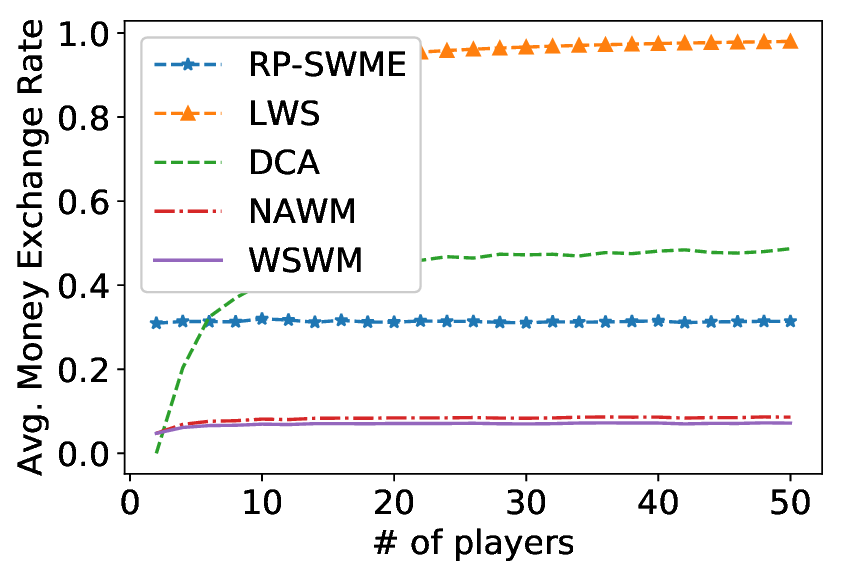}
\caption{$\mathbf{p}\sim$ Logit, $\mathbf{w}\sim$ Uniform}
\label{me_u_l}
\end{subfigure}
\begin{subfigure}{0.32\textwidth}
\includegraphics[width=\textwidth]{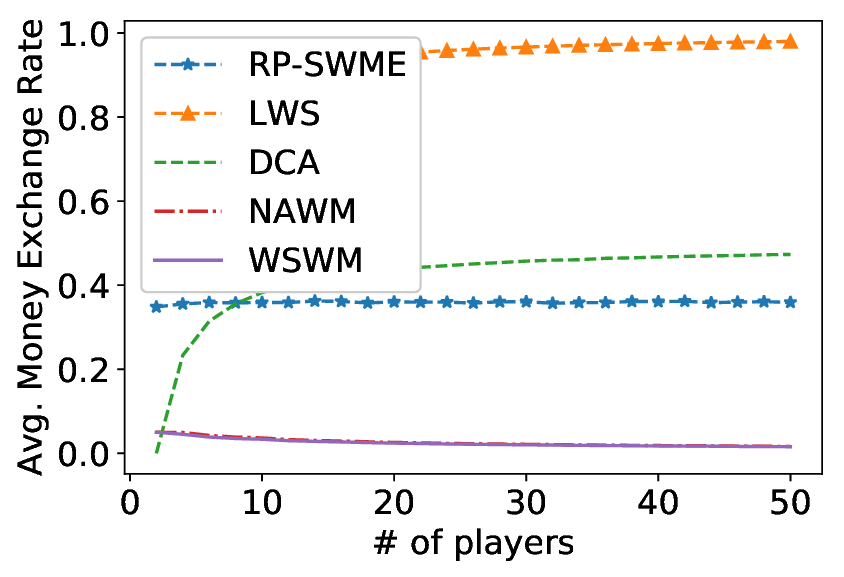}
\caption{$\mathbf{p}\sim$ Synthetic, $\mathbf{w}\sim$ Uniform}
\label{me_u_s}
\end{subfigure}
\begin{subfigure}{0.32\textwidth}
\includegraphics[width=\textwidth]{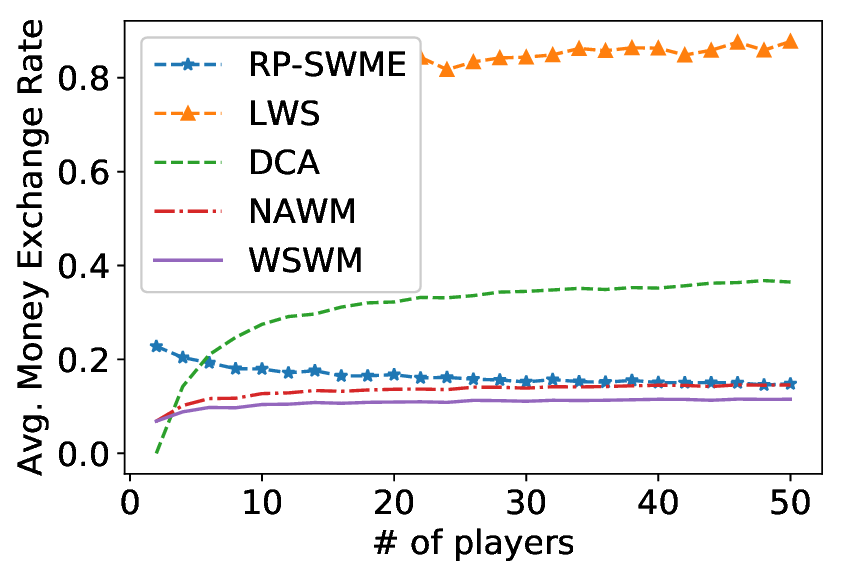}
\caption{$\mathbf{p}\sim$ Uniform, $\mathbf{w}\sim$ Pareto}
\label{me_p_u}
\end{subfigure}
\begin{subfigure}{0.32\textwidth}
\includegraphics[width=\textwidth]{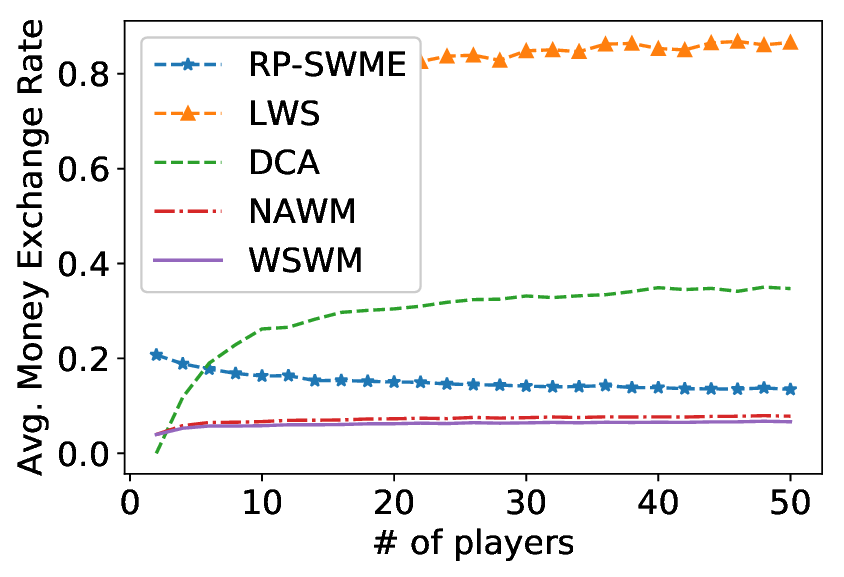}
\caption{$\mathbf{p}\sim$ Logit, $\mathbf{w}\sim$ Pareto}
\label{me_p_l}
\end{subfigure}
\begin{subfigure}{0.32\textwidth}
\includegraphics[width=\textwidth]{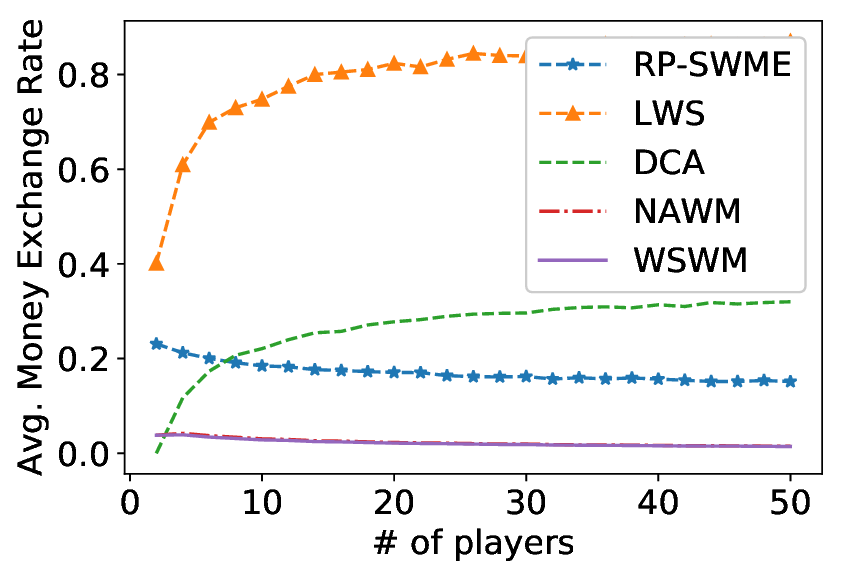}
\caption{$\mathbf{p}\sim$ Synthetic, $\mathbf{w}\sim$ Pareto}
\label{me_p_s}
\end{subfigure}
\caption{\label{me_all}Average money exchange rate of each of five wagering mechanisms as a function of $N$ under different prediction and wager models}
\end{figure}
\begin{figure}[t]
\begin{subfigure}{0.32\textwidth}
\includegraphics[width=\textwidth]{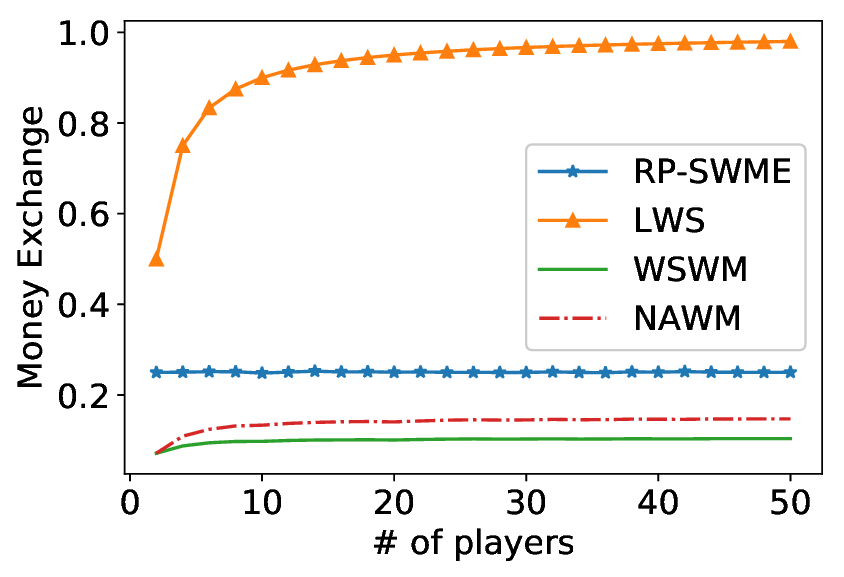}
\caption{\label{me_multi_3}Events with 3 outcomes}
\end{subfigure}
\begin{subfigure}{0.32\textwidth}
\includegraphics[width=\textwidth]{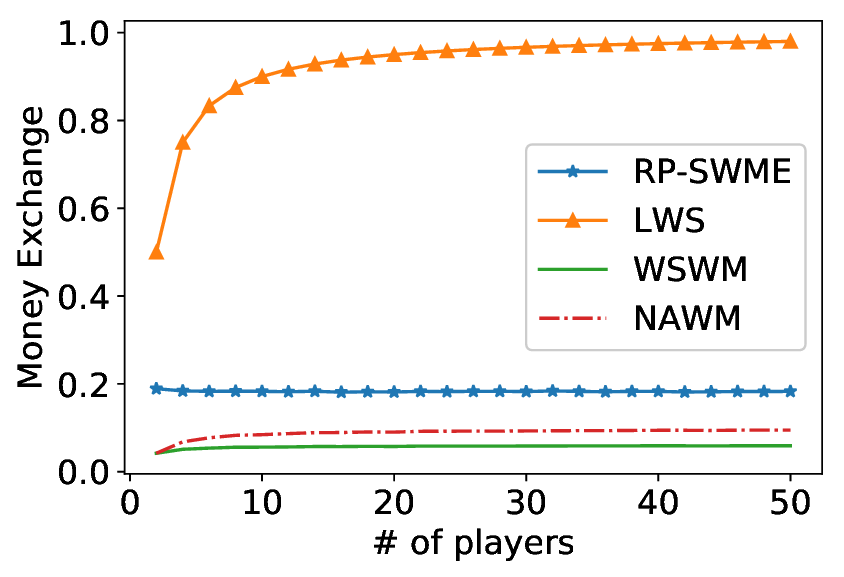}
\caption{\label{me_multi_6}Events with 6 outcomes}
\end{subfigure}
\begin{subfigure}{0.32\textwidth}
\includegraphics[width=\textwidth]{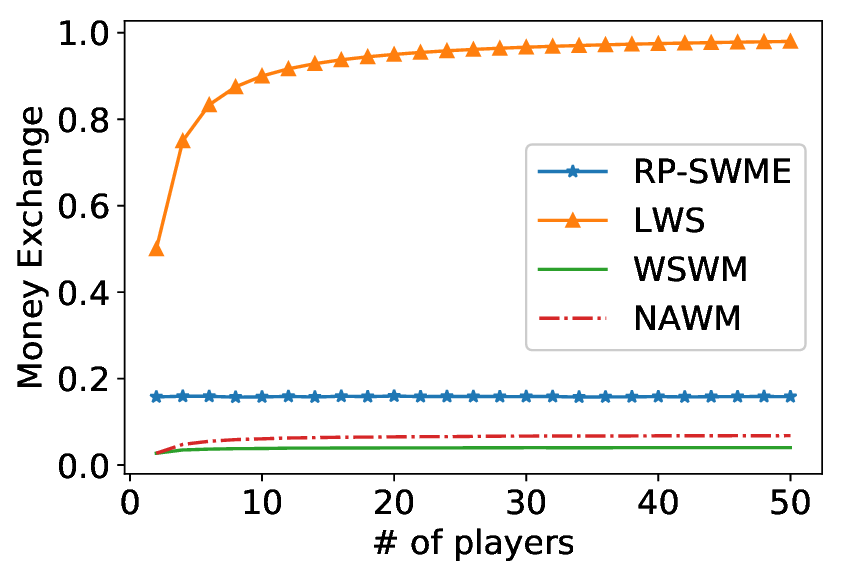}
\caption{\label{me_multi_9}Events with 9 outcomes}
\end{subfigure}
\caption{\label{me_multi_all}Average money exchange rate of each of four mechanisms under events with multiple outcomes}
\end{figure}

\textsf{LWS}  doubles the money exchange rate of the second best alternative, hitting a more than 80\% money exchange rate under all conditions we simulated (Figure~\ref{me_all}, \ref{me_multi_all}). On the other hand, \textsf{RP-SWME} also defeats the other two incentive compatible deterministic wagering mechanisms in expected money exchange under all conditions we simulated  (Figure~\ref{me_all}, \ref{me_multi_all}). Meanwhile, it also outperforms \textsf{DCA} when the number of agents is small  (Figure~\ref{me_all}).

In particular, when the prediction follows the synthetic model, where the predictions are much closer to each other as  the number of participants increases,  the money exchange rate of the two incentive compatible deterministic wagering mechanisms,  \textsf{WSWM} and  \textsf{NAWM}  converge to zero. However, the two randomized wagering mechanisms still keep a large money exchange rate (Figure~\ref{me_u_s}, \ref{me_p_s}).

\subsection{Comparison of randomness properties of  \textsf{RP-SWME} and  \textsf{LWS}}
\label{sim_var_pr}
In this section, we compare the \emph{standard variance} of payoffs and \emph{the probability of not losing money} of \textsf{RP-SWME} and \textsf{LWS}. We evaluated these two metrics w.r.t. to the prediction accuracy, which is measured based on the distance of a prediction to the outcome, i.e., $\text{Accuracy} = 1 - |x-p_i|$\footnote{We use it as measurement of accuracy for two reasons: i. it is linear in prediction $p_i$, ii. it has an inject to Brier Score}.

In the evaluation, we run 10000 wagering instances under these two mechanisms and recorded the prediction accuracy of each agent in each instance and the corresponding net-payoff. Then, we group these agents into 10 groups that correspond to 10 consecutive accuracy intervals. In each group, we calculate the standard variance and the percent of agents winning money. For fair comparison, we normalize the net-payoff of each agent by its own wager.

We simulate binary events. We generate two set of simulated data. In both sets, we varied the number of agents from 2 to 50 with a set of 2, and under each number, we generated 10000 instances. In each instance, the agents' predictions are drawn from the Uniform model, while the wagers are drawn from the Uniform model in one set and drawn from the Pareto model in the other set. 

Our results show that under all conditions we simulate, \textsf{RP-SWME} has a much smaller variance in agents' net-payoff and the variance is steady across agents with different prediction accuracy. In contrast, the \textsf{LWS} has a much larger variance in net-payoff, which increases with the prediction accuracy (Figure~\ref{std_all}). On the other hand,  \textsf{RP-SWME} has a much larger probability of not losing money and this probability increases with the prediction accuracy, while \textsf{LWS} has a much smaller such probability (Figure~\ref{pr_all}). In brief, while both \textsf{RP-SWME} and \textsf{LWS} can effectively improve the efficiency of wagering, \textsf{RP-SWME} provides much less uncertainty than \textsf{LWS} does and thus, may be regarded as a more attractive alternative for deterministic wagering mechanisms. 

\begin{figure}[!t]
\begin{minipage}{0.50\textwidth}
\begin{subfigure}{0.46\textwidth}
\includegraphics[width=\textwidth]{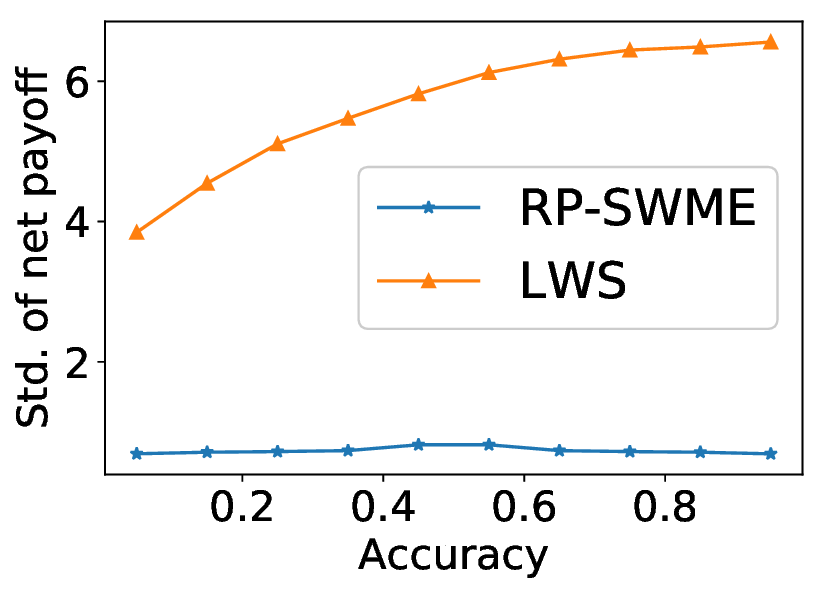}
\caption{\label{std_o_uniform}
$\mathbf{w}\sim$ Uniform}
\end{subfigure}
\begin{subfigure}{0.48\textwidth}
\includegraphics[width=\textwidth]{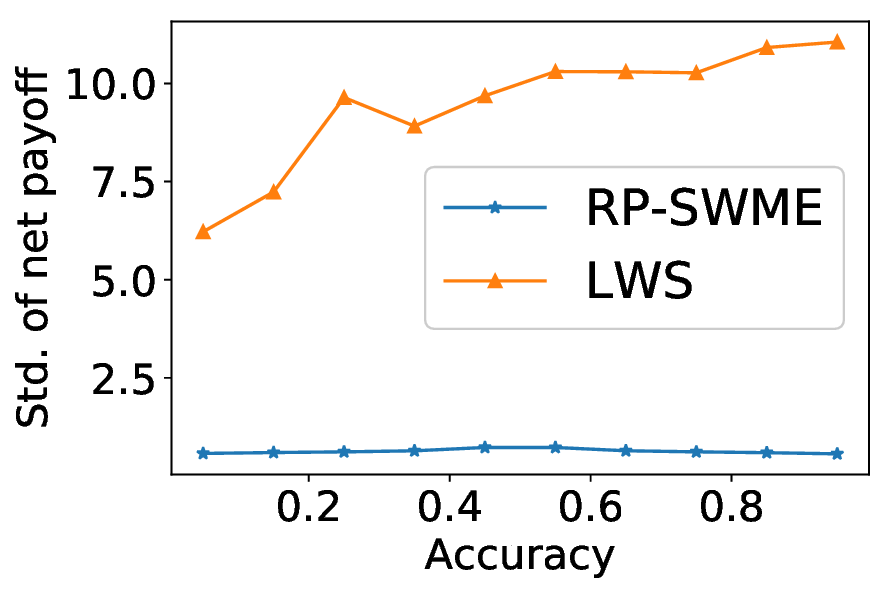}
\caption{\label{std_o_pareto}
$\mathbf{w}\sim$ Pareto}
\end{subfigure}
\caption{\label{std_all}Std. variance of net-payoff as a function of prediction accuracy: \textsf{RP-SWME} v.s. \textsf{LWS}}
\end{minipage}
\begin{minipage}{0.50\textwidth}
\begin{subfigure}{0.46\textwidth}
\includegraphics[width=\textwidth]{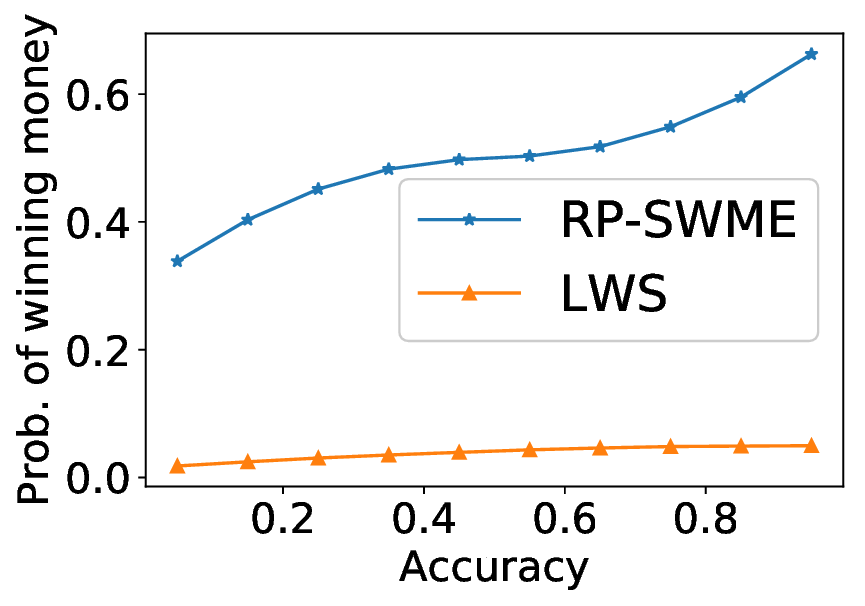}
\caption{\label{pr_o_uniform}
$\mathbf{w}\sim$ Uniform}
\end{subfigure}
\begin{subfigure}{0.48\textwidth}
\includegraphics[width=\textwidth]{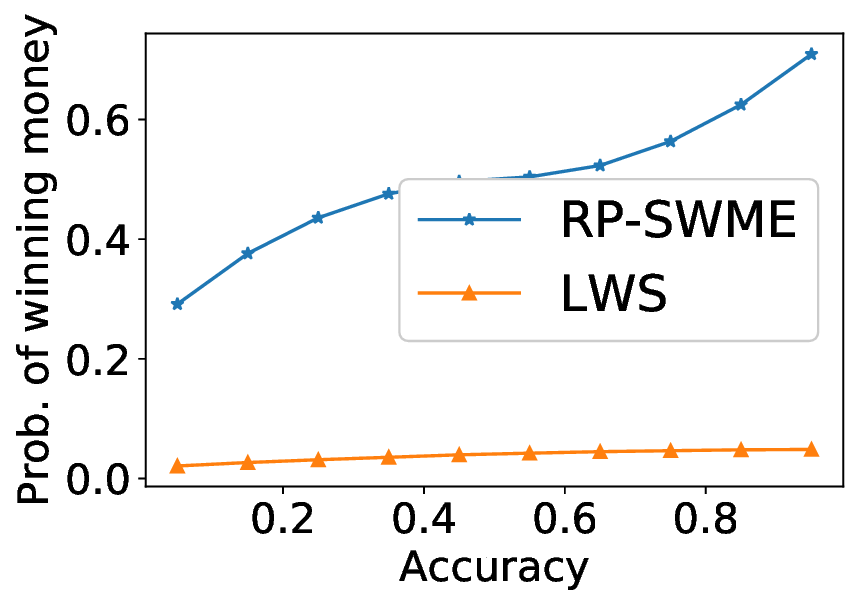}
\caption{\label{pr_o_pareto}
$\mathbf{w}\sim$ Pareto}
\end{subfigure}
\caption{\label{pr_all}Probability of winning money as a function of prediction accuracy: \textsf{RP-SWME} v.s. \textsf{LWS}}
\end{minipage}
\end{figure}

\section{Conclusion}\label{sec:conclude}
We extend the design of wagering mechanism to its randomized space. We propose two of them: Lottery Wagering Mechanisms (\textsf{LWM}) and Surrogate Wagering Mechanisms (\textsf{SWM}). We demonstrate the power of randomness by theoretically proving that they both satisfy a set of desirable properties, including Pareto efficiency which is missing in exiting wagering literature. We also carried out extensive experiments to support our theoretical findings. \textsf{SWM} is also robust to noisy outcomes. In particular, as shown by simulations, surrogate wagering mechanisms have reasonably small standard variance in agents' payoff and low probability for agents to lose all their wagers.

\bibliography{myref,library}


\end{document}